\documentclass[11pt]{article}
\usepackage{fullpage}

\usepackage{latexsym}
\usepackage{amsmath}
\usepackage{amssymb}
\usepackage{amsthm}
\usepackage{hyperref}
\usepackage{cite}
\usepackage{graphicx}
\usepackage{color}
\usepackage{enumerate}

\usepackage{tikz}
\usetikzlibrary{arrows,decorations.pathmorphing,backgrounds,positioning,fit,petri}
\usepackage{caption}
\usepackage{subcaption}
\usepackage{float}

\newtheorem{theorem}{Theorem}
\newtheorem{theorem*}{Theorem}

\newtheorem{lemma}[theorem]{Lemma}

\newtheorem{proposition}[theorem]{Proposition}

\theoremstyle{definition}

\newcommand{\F}{\mathbb{F}}

\newcommand{\ctabc}{\mathcal{C}}

\newcommand{\Dom}{\mathrm{Dom}}

\newcommand{\inp}{\mathrm{in}}
\newcommand{\out}{\mathrm{out}}
\newcommand{\mi}{\mathrm{middle}}

\newcommand{\mdeg}{\mathrm{mdeg}}
\newcommand{\cdeg}{\mathrm{cdeg}}

\newcommand{\Pol}{\mathrm{P}}
\newcommand{\NP}{\mathrm{NP}}
\newcommand{\coNP}{\mathrm{coNP}}

\newcommand{\cS}{\mathcal{S}}

\newcommand{\PPA}{\mathrm{PPA}}
\newcommand{\PPAD}{\mathrm{PPAD}}

\newcommand{\Pe}{\mathrm{P}}

\newcommand{\suppress}[1]{}

\newcommand{\miklos}[1]{\textcolor{green}{\small$\bullet$\footnote{\textcolor{green}{Miklos:
 #1}}}}

\newcommand{\alex}[1]{\textcolor{violet}{\small$\bullet$\footnote{\textcolor{violet}{Alex: #1}}}}

\title{On the Polynomial Parity Argument complexity of the Combinatorial Nullstellensatz
}

\author{
Aleksandrs Belovs\thanks{Faculty of Computing, University of Latvia ({\tt stiboh@gmail.com}).}
\and
G\'abor Ivanyos\thanks{Institute for Computer Science and Control, Hungarian Academy of Sciences, 
Budapest, Hungary 
({\tt Gabor.Ivanyos@sztaki.mta.hu}).} 
\and
Youming Qiao\thanks{Centre for Quantum Computation and Intelligent Systems, 
 University of Technology Sydney, Australia  ({\tt jimmyqiao86@gmail.com}).}
\and 
Miklos Santha\thanks{IRIF, Universit\'e Paris Diderot, CNRS, 75205 Paris, France;  and
Centre for Quantum Technologies, National University of Singapore and MajuLab, CNRS,
Singapore 117543 ({\tt santha@irif.fr}).}
\and 
Siyi Yang\thanks{
Centre for Quantum Technologies, National University of Singapore,
Singapore 117543 ({\tt ysyshtc@gmail.com}).}
}

\begin{document}
\maketitle

\begin{abstract}
The 
complexity class PPA consists of NP-search problems which are reducible to
the parity principle in undirected graphs. 
It contains a wide variety of interesting problems from graph theory,
combinatorics, algebra and number theory, but only a few of these are known to be 
complete in the class. 
Before this work, the known complete problems were all discretizations or 
combinatorial analogues of topological fixed point theorems.

Here we prove the PPA-completeness of two  problems of radically different style.
They are \textsc{PPA-Circuit CNSS} and \textsc{PPA-Circuit Chevalley}, related respectively to the Combinatorial Nullstellensatz 
and to the Chevalley-Warning Theorem over the two elements field $\F_2$.
The input of these problems contain PPA-{\em circuits}  which are arithmetic
circuits with special symmetric properties that assure %
that the polynomials computed by them have always an even 
number of zeros.
In the proof of the result we relate the multilinear degree of the polynomials to the parity of the {\em maximal parse subcircuits} that compute monomials
with maximal multilinear degree, and we show that the maximal parse subcircuits of a PPA-circuit can be paired
in polynomial time.

\end{abstract}

\newpage

\section{Introduction}
\label{sec:intro}

\paragraph{The class PPA.} The complexity class TFNP~\cite{MP91} consists of 
$\NP$-search problems corresponding to total relations.
In the last 25 years various subclasses of TFNP have been thoroughly investigated. The polynomial
parity argument classes $\PPA$ and $\PPAD$ were defined in the seminal work of Papadimitriou~\cite{Pap90}.
$\PPA$ consists of the search problems which are reducible to the parity principle stating that
in an undirected  graph the number of odd vertices is even. The more restricted class $\PPAD$ is based on the analogous
principle for directed graphs.

The class $\PPAD$ contains a relatively large number of complete problems from various areas of mathematics.
In his paper Papadimitrou~\cite{Pap90}
has already shown that among others the 3-dimensional \textsc{Sperner}, \textsc{Brouwer}
problems, as well as the \textsc{Exchange Equilibrium} problem from mathematical economics were $\PPAD$-complete.
A few years later Chen and Deng~\cite{CD06} proved 
 that 2-dimensional \textsc{Sperner} was also $\PPAD$-complete,
and after a sequence of beautiful papers Chen and Deng~\cite{CDb06} has established 
the $\PPAD$-completeness of computing 2-player Nash equilibrium, see also~\cite{CDT09}.
Kintali~\cite{Kin} has compiled a list of 25 $\PPAD$-complete problems; 
the list is far from complete.

In comparison with $\PPAD$, relatively few complete problems are known in the class $\PPA$, all of which are
discretizations or combinatorial analogues  of topological fixed point theorems.
While the original paper of Papadimitriou~\cite{Pap90} exhibited a large collection of problems in $\PPA$, 
none of them was proven to be PPA-complete.
Historically the first PPA-completeness result was given by Grigni~\cite{Gri01}
who,
realizing that analogues of $\PPAD$-complete problems in non-orientable spaces could 
become PPA-complete, has shown the PPA-completeness of 
the \textsc{Sperner} problem for a non-orientable 
3-dimensional space.
This result was strengthened by Friedl et al.~\cite{KISV06} to a non-orientable and locally 2-dimensional space.
Up to 
 our knowledge, until 2015 just these two problems were known to be PPA-complete.
Last year Deng et al.~\cite{DEFLQX16} established the PPA-completeness of several 2-dimensional problems 
on the M\"obius band,
including \textsc{Sperner} and \textsc{Tucker}, and they have obtained similar results 
for the Klein bottle and the projective plane.
Recently Aisenberg, Bonet and Buss~\cite{ABB15} have shown that 
2-dimensional \textsc{Tucker} in the Euclidean space was PPA-complete.

Compared to the fundamental similarity of these complete problems in PPA, the list of problems in the class
for which no completeness result is known is very rich. Already in Papadimitriou's paper~\cite{Pap90}
we find problems from 
graph theory, such as \textsc{Smith} and \textsc{Hamiltonian decomposition}, from
combinatorics, such as \textsc{Necklace splitting} and \textsc{Discrete Ham 
sandwich} (the proof in~\cite{Pap94} that these problems are in $\PPAD$ was 
incorrect~\cite{ABB15}), and from algebra,
a variant of Chevalley's theorem over the 2 elements field $\F_2$, which we call \textsc{Explicit Chevalley}.
Cameron and Edmonds~\cite{CE99} gave new proofs based on the parity principle for a long series of theorems 
from graph theory~\cite{Toi73, Wes78, BH86, Ber86, CE90}, the corresponding search problems are therefore in PPA.
Recently Je\v{r}\'abek~\cite{Jer16} has put several number theoretic problems, such as
square root computation  and finding quadratic nonresidues modulo $n$ into PPA, and he has also shown that
\textsc{Factoring}  is in PPA under randomized reduction. 

\paragraph{Our contribution.} The main result of this paper is that two 
appropriately defined problems
related to Chevalley-Warning Theorem~\cite{Ch36, W36} and to Alon's Combinatorial Nullstellensatz~\cite{Alon99}
over $\F_2$
are complete in PPA. 
These are 
the first PPA-completeness results involving problems which are not
inspired by topological fixed point theorems.

The Chevalley-Warning Theorem is a classical result about zeros of polynomials.
It says that if $P_1, \ldots , P_k$ are $n$-variate polynomials over a field of
characteristic $p$ such that the sum of their degrees is less than $n$, then the number of common zeros is 
divisible by $p$.
The Combinatorial Nullstellensatz (CNSS) of Alon  states that if $P$ is an $n$-variate polynomial over $\F$
whose degree is 
$d_1 + \cdots + d_n$, and this is certified by the monomial
$c x_1^{d_1} \cdots x_n^{d_n}$, for some $c\neq 0$, then in $S_1 \times \cdots \times S_n \subseteq \F^n$ there exists a point
where $P$ is not zero, whenever $|S_i| > d_i$, for $i = 1, \ldots, n$. The CNSS has found a wide range of applications
among others in graph theory, combinatorics and additive number theory~\cite{Alon99, Alon02}. 

Over the field $\F_2$
the two theorems greatly simplify via the notion of {\em multilinear degree}. For any polynomial $P$ over $\F_2$, there exists
a unique multilinear polynomial $M$ such that $P$ and $M$ compute the same
function on $\F_2^n$. 
We call the degree of $M$ the
multilinear degree of $P$, denoted as $\mdeg(P)$. We  
use $\deg(P)$ to denote the usual degree of $P$. Then the Chevalley-Warning 
Theorem and the CNSS over $\F_2$  are equivalent to the following 
statement:
\begin{center}
{\em An $n$-variate $\F_2$-polynomial  has an odd number of zeros if and only if its multilinear degree is $n$.}
\end{center}
The natural search problem corresponding to the CNSS therefore is: given an $n$-variate  
polynomial $P$ whose multilinear degree is $n$, find
a point $a$ where $P(a)=1$. Similarly, the search problem corresponding to the Chevalley-Warning Theorem is: 
given an $n$-variate polynomial $P$ whose multilinear degree is less than $n$ and a zero of $P$, find another zero.

\suppress{
Obviously, these problems are not yet well defined algorithmically, since it is not specified, how the polynomial $P$ is given.
The starting point of our investigations is the result of Papadimitriou about the instantiation of the Chevalley-Warning
Theorem.
Let the polynomial $P+1$ be represented
by a $\Pi_3$-arithmetic formula, that is a formula which is the product of sums of products. In other words,
$P+1$ is the product of explicitly given polynomials $P_1+1, \ldots, P_k+1$. 
Clearly, $P(x) =0$ if and only if $P_i(x) =0$, for $i \in [n]$.
\miklos{In the previous sentences I added 1 to the polynomials, this was previously incorrect. 
The next sentence is new}
In that case the degree of the arithmetic formula is the same as the degree of $P$, they are both equal to the 
sum of the degrees of the polynomials $P_1, \ldots, P_k$.
Suppose that this degree is less then $n$,
and suppose that we are given $a \in \F_2^n$ such that $P(a) = 0$.
Then the task is to find $a' \neq a$ such that $P(a') = 0$. We call this problem \textsc{Explicit Chevalley}, and
Papadimitriou has shown~\cite{Pap90} that it is in PPA.
}

Obviously, these problems are not yet well defined algorithmically, since it is not specified, how the polynomial $P$ is given.
The starting point of our investigations is the result of Papadimitriou about some instantiation of the Chevalley-Warning
Theorem. Specifically, in \cite{Pap90} Papadimitriou considered the following 
problem. 
Let the polynomials $P_1, \ldots, P_k$ be given explicitly as sums of monomials,
and define $P(x) = 1+  \prod_{i=1}^k (P_i(x) + 1)$. 
We have then $\deg(P) = \sum_{i=1}^k \deg (P_i)$, and 
clearly $P(x) =0$ if and only if $P_i(x) =0$, for $i \in [n]$.
Suppose that $\deg (P) < n$, and that 
we are given $a \in \F_2^n$ such that $P(a) = 0$.
Then the task is to find $a' \neq a$ such that $P(a') = 0$. We call this problem \textsc{Explicit Chevalley}, and
Papadimitriou has shown~\cite{Pap90} that it is in PPA.

Could it be that \textsc{Explicit Chevalley} is PPA-complete? We find this highly unlikely. There are
two restrictions on the input of \textsc{Explicit Chevalley}. 
Firstly, the polynomial $P$ is given by an arithmetic 
circuit (in fact by an arithmetic formula) of specific form.
Secondly, 
the number of variables not only upper bounds the multilinear degree of $P$, but also
the degree of $P$.
The first restriction can be easily  relaxed.
We can define and compute recursively very easily 
the circuit degree (also known as the formal degree; see Section~\ref{subsec:ckt}) 
of 
the 
arithmetic 
circuit which is an upper bound on the degree of the polynomial computed by the circuit.
Could it be that the problem, specified by an arithmetic circuit whose circuit degree 
is less than $n$,  becomes PPA-complete?
While this problem might be indeed harder than \textsc{Explicit Chevalley}, we still don't think that it is PPA-complete.

We believe that the more important restriction in Papadimitriou's problem is the one on the degree of the polynomial $P$
computed by the input circuit.
As we have seen, to have an even number of zeros, 
mathematically it is only required that the multilinear degree of $P$ is less than $n$, so putting the restriction
on the degree of $P$
is too stringent. Let's try then to consider instances specified by arithmetic circuits computing polynomials of 
multilinear degree less than $n$. 
However, here we face a serious difficulty. 
We can't just promise that the polynomial has multilinear degree less than $n$ since $\PPA$ 
is a syntactic class.
We must be able to verify syntactically that it is indeed the case. 


The multilinear degree of the polynomial is decided by the parity of the monomials computed by the circuit 
which contain every variable.
Let us call such monomials \emph{maximal}.
Indeed, the multilinear degree of $P$ is less than $n$ if and only if an even number
of maximal monomials are computed by the circuit. 
A very general way to prove efficiently that a set is of even cardinality is to give a 
polynomial Turing machine which computes a perfect matching on the elements of the set. However, 
the parsing of monomials in arbitrary arithmetic circuits is a rather complex task~\cite{Mal03}.
For a start, the number of maximal monomials computed by a polynomial size arithmetic
circuit 
can be doubly exponential,
making even the description of such a monomial impossible in polynomial time. Fortunately, the situation over the field $\F_2$
simplifies a lot, thanks to cancellations due to certain symmetries.
In fact, we are able to show that over $\F_2$ it is sufficient to consider only
those monomials which are computed by consistent left/right labellings of the sum gates 
participating in the computation of the monomial,
because the rest of the monomials cancel out. 
We call such labellings
{\em parse subcircuits}, and we call those parse subcircuits
which compute maximal monomials {\em maximal}. 
The introduction of parse subcircuits was 
inspired by the concept of parse trees in~\cite{JS82, MP08}.
Technically, this results shows that that computing the multilinear degree is in 
$\oplus \Pe$, the complexity class Parity P.


Is there a chance that for a general circuit computing the multilinear degree is in P?
As it turns out not, unless $\oplus \Pe = \Pe$, because we can show that computing the multilinear degree is also 
$\oplus \Pe$-hard.
Therefore we have to identify a restricted class of circuits computing polynomials of even multilinear degree which satisfy two properties:
the class is on the one hand restricted enough that we are able to construct a polynomial time perfect matching 
for the maximal parse subcircuits, but it is also large enough that finding another zero for the circuit is PPA-hard. 
The main contribution of this paper is that we identify such a class of arithmetic circuit which we call {\em $\PPA$-circuits}.

The definition of these circuits is inspired by a rather straightforward 
translation of  Papadimitriou's basic
$\PPA$-problem into a problem for arithmetic circuits. In a nutshell,
the basic $\PPA$-problem is the following. Given a degree-one vertex of 
a graph, in which every vertex has degree at most two, find another
degree-one vertex.  Here, the graph, whose vertices are the $0$-$1$ strings
of given length, is 
given via a
polynomial time Turing machine $M$ determining the neighbourhood
of any specified node. We construct an arithmetic circuit
over $\F_2$ which, given a vertex $v$ in this graph, computes 
the opposite parity of the number of $v$'s neighbours. 
Therefore, finding another degree-one vertex is then just the same
as finding another zero of the polynomial computed by the circuit.  
Most importantly, the circuit is constructed to be in a special form, which 
allows for a polynomial-time-computable perfect 
matching over its maximal parse subcircuits. 
Roughly speaking, from the Turing machine $M$ that describes the neighbours of 
vertices, we extract two arithmetic circuits $D$ and $F$ that also describe 
the neighbours in a certain way. We then define the so-called PPA-{\em 
composition} of these two circuits, which 
produces 
a circuit $C$ that accesses $D$ 
and $F$ in a black box fashion.
Symmetries of the PPA-composition, reflecting the special structure of degree
computation, enable us to construct a polynomial-time-computable perfect 
matching over its maximal parse subcircuits (cf. Lemma~\ref{lem:matching}).
Finally we define a PPA-{\em circuit} as the sum of a PPA-composition and another 
circuit whose
circuit degree is less than $n$. This is just a minor extension of the family of 
PPA-compositions since circuits with degree less than $n$
don't have maximal parse subcircuits. The reason for considering this extended 
family is that this way our result immediately generalizes
Papdimitrou's result~\cite{Pap90} about \textsc{Explicit Chevalley}, and it makes 
also easier to express the equivalence 
between the algorithmic versions of the Chevalley-Warning theorem and the CNSS.

The definition of our two problems, \textsc{PPA-Circuit-CNSS} and \textsc{PPA-Circuit-Chevalley}, is therefore the following.
In both cases we are given an $n$-variable, PPA-circuit $C$ over $\F_2$ and an element $a \in \F_2^n$.
In the case of \textsc{PPA-Circuit Chevalley}, $a$ is a zero of $C$, and 
for \textsc{PPA-Circuit CNSS}, we consider the sum of the circuits $C$ and $L_a$, where $L_a$ is
a simple {\em Lagrange-circuit} having $a$ as its only zero and having a single
maximal parse subcircuit.
The computational task is to compute another zero of $C$ in case of \textsc{PPA-Circuit Chevalley}, and a satisfying
assignment for $C + L_a$ in case of \textsc{PPA-Circuit CNSS}.
Our result is then stated in the following
theorem.
\begin{theorem}
\label{thm:main}
The problems \textsc{PPA-Circuit CNSS} and \textsc{PPA-Circuit Chevalley} are $\PPA$-complete.
\end{theorem}

Since the two problems are easily interreducible, for the proof of
Theorem~\ref{thm:main} we will show that \textsc{PPA-Circuit CNSS} is PPA-easy and \textsc{PPA-Circuit Chevalley} is 
PPA-hard. 
For the easiness part we define a graph,
inspired by Papadimitriou's construction,
whose vertices are the assignments for the variables and the parse subcircuits. There is an edge
between a parse subcircuit and an assignment if the monomial defined by the subcircuit takes the value 1 on
the assignment.
In addition, we also put an edge between two maximal parse subcircuits of the PPA-composition part of the circuit
if they are paired by the perfect matching.
As it turns out, the odd degree vertices in this graph are exactly the assignments where the polynomial defined by the circuit
is 1, and the unique maximal parse subcircuit of the Lagrange-circuit.
Technically, the main part of the proof is to give,
for every assignment, a polynomial time computable pairing between its exponentially many 
neighboring parse subcircuits. For the hardness part (which is much simpler to prove) we express the
basic PPA-complete problem as a PPA-composition, as we explained above.

\paragraph{Previous work.} 
Papadimitriou has proven that \textsc{Explicit Chevalley}
is in PPA.
Varga~\cite{Varga14} has shown the same for 
the special case of \textsc{CNSS} 
where the input polynomial $P$ is specified as the sum of a polynomial number of polynomials $P_i$, 
where each $P_i$ is the product of explicitly given polynomials whose sum of degrees is at most $n$.
In addition, the input also contains a polynomial time computable matching for all but one of the monomials $x_1 \cdots x_n$ of $P$.
However, the paper doesn't address the question why this doesn't make the problem a promise problem. 
Concerning the hardness of \textsc{CNSS}, Alon proved in ~\cite{Alon02} the 
following result. Let $P$ be
specified by an arithmetic circuit in a way that it can be checked efficiently that its multilinear degree is $n$. If a polynomial
time algorithm can find a point $a$ where $P(a) = 1$, then there are no one-way permutations.

\paragraph{Structure of the paper.} 
In {\bf Section~\ref{sec:prelim}} we recall the definition of the class PPA, the 
Combinatorial Nullstellensatz and
the Chevalley-Warning Theorem, and arithmetic circuits. In {\bf 
Section~\ref{sec:parse_sub}} we define the
parse subcircuits of an arithmetic circuit over $\F_2$, and in {\bf Proposition~\ref{prop:char2}} we prove that the polynomial
computed by the circuit is the sum of the monomials computed by the 
parse 
subcircuits. In {\bf Section~\ref{sec:ppa-circuit}} we define PPA-circuits, and in
Lemma~\ref{lem:matching} we prove that in such circuits 
a perfect matching for the maximal parse subcircuits can be computed in polynomial time.
In {\bf Section~\ref{sec:problems}}
we state the problems \textsc{PPA-Circuit CNSS} and \textsc{PPA-Circuit Chevalley} over $\F_2$ and observe
that they are polynomially interreducible. In {\bf Section~\ref{sec:easiness}} in 
{\bf Theorem~\ref{thm:easiness}} we prove that \textsc{PPA-Circuit CNSS} is in PPA, and in
{\bf Section~\ref{sec:hardness}} in {\bf Theorem~\ref{thm:hardness}} we prove that \textsc{PPA-Circuit Chevalley} is PPA-hard.

\section{Preliminaries}
\label{sec:prelim}
\subsection{Total functional $\NP$ and the class $\PPA$}
\label{sec:prelim_ppa}
We denote the set $\{1, \ldots, n\}$ by $[n]$. 
A polynomially computable binary relation $R \subseteq \{0,1\}^* \times \{0,1\}^*$ is called balanced if
for some polynomial $p(n)$, for every $x$ and $y$ such that $R(x,y)$ holds, we
have $|y|\leq p(|x|)$. Such a relation defines an {\em $\NP$-search problem} $\Pi_R$ whose input is $x$, 
and the task is to find for inputs $x$, where $R(x,y)$ holds for some $y$, such a {\em solution} $y$,
and report ``failure" otherwise.
The class FNP of {\em functional} NP consists of NP-search problems.
For two problems $\Pi_R$ and $\Pi_S$ in FNP, we say that $\Pi_R$ is {\em reducible to}
$\Pi_S$ if there exist two functions $f$ and $g$ computable in polynomial time such that for every positive $x$, 
$S(f(x),y)$ implies $R(x,g(x,y))$. 

An NP-search problem is
{\em total} if for every $x$, there exists a solution $y$. The class of these problems is
called TFNP (for Total Functional NP) by Megiddo and Papadimitriou~\cite{MP91}.
Problems in TFNP exhibit very interesting complexity properties. An FNP-complete
search problem can not be total unless $\NP = \coNP$. It is also unlikely that every problem in TFNP could be solved
in polynomial time since this would imply $\Pol = \NP \cap \coNP$.
TFNP is a semantic complexity class, in the sense that it involves a promise about the totality of the relation $R$.
It is widely believed that such a promise can not be enforced syntactically on a Turing machine, 
in fact there is no known recursive enumeration
of Turing machines that compute total search problems. As usual with semantic complexity classes, TFNP 
doesn't seem to have complete problems. On the other hand, several syntactically defined subclasses of
TFNP  with a rich structure of complete problems have been identified along the lines of the mathematical proofs establishing the 
totality of the defining relation.

The parity argument subclasses of TFNP were defined by Papadimitriou~\cite{Pap90, Pap94}.
They
can be specified via concrete problems, by closure under reduction. The
\textsc{Leaf} problem is defined as follows. The input is a triple $(z,M,\omega)$ where $z$ is a binary string and 
$M$ is the description of a
polynomial time Turing machine\footnote{The requirement for $M$ to run in 
polynomial can be imposed by adding a clock.}
that defines a graph $G_z = (V_z,E_z)$
as follows. The set of vertices 
is $V_z = \{0,1\}^{p(|z|)}$ for some polynomial $p$.
For any vertex $v \in V_z$, the machine $M$ outputs on $(z,v)$ a set of at most two vertices. 
Then, we define $G_z$ as a graph without self-loops, where  $\{v, v'\} \in E_z$ for $v \neq v'$, 
if $v' \in M(z,v)$ and $v \in M(z, v')$. 
Obviously $G_z$ is an undirected graph where the degree of each vertex is at most 2, and therefore the number of 
leaves, that is of degree one vertices, is even. Finally $\omega \in V_z$ is a degree one vertex that we call the {\em standard leaf}.
The output of the problem \textsc{Leaf} is a leaf of $G_z$ different from the standard leaf.
The Polynomial Parity Argument class PPA is the set of total search problems reducible to  \textsc{Leaf}.
The directed class PPAD is defined by \textsc{D-Leaf}, the directed analog of \textsc{Leaf}.
In the problem \textsc{D-Leaf} the Turing machine defines a directed graph, where the indegree and outdegree of every
vertex is at most one. 
The standard leaf
$\omega$ is
a source, and the output is a sink or
source different from the $\omega$. 

As shown in \cite{Pap94}, the definition of $\PPA$ can capture also those problems 
for which the underlying 
graph has unbounded degrees and we are seeking for another odd-degree vertex.
 Specifically, suppose there exists a
polynomial time {\em edge recognition} algorithm $\epsilon(v,v')$, 
which 
decides whether $\{v,v'\} \in E_z$. 
Assume also, that in addition we have a polynomial time {\em pairing function} 
$\phi(v,w)$, where by definition, for every vertex $v$, 
the function $\phi(v, \cdot)$ 
satisfies the following properties.
For every even degree vertex $v$, 
it is a pairing between the vertices adjacent to $v$, that is 
for every such vertex $w$, 
we have $\phi(v,w) = w'$,
where $w' \neq w$, $w'$ is also adjacent to $v$, and $\phi(v,w') = w$. For odd 
degree vertices $v$, we have exactly one 
adjacent vertex $w$ such that $w$ is mapped to itself, and on the remaining 
adjacent vertices it is pairing
as in the case of an even degree vertex $v$. The input also contains an odd degree 
vertex $v$ with a proof for that, in the form
of an adjacent vertex $w$, such that $\phi(v,w) = w$. 
In \cite[Corollary to 
Theorem~1]{Pap94}, Papadimitriou showed that any problem defined in terms of an 
edge recognition algorithm and a pairing 
function is in $\PPA$.


\subsection{Combinatorial Nullstellensatz and Chevalley-Warning Theorem}

Let $\F$ be a field. An {\em polynomial over $\F$}
(or shortly a polynomial) in $n$ variables is a formal expression 
$P(x) = P(x_1, \ldots, x_n)$ of the form
$$
P(x_1, \ldots, x_n) = \sum_{d_1, \ldots , d_n \geq 0} c_{d_1, \ldots , d_n}x_1^{d_1} \cdots x_n^{d_n},
$$
where the coefficients $c_{d_1, \ldots , d_n}$ are from $\F$, and only a finite 
number of them are different from zero.
The {\em degree} $\deg(P)$ of $P$ is the largest value of $d_1 +  \cdots + d_n $ for which the coefficient $c_{d_1, \ldots , d_n}$
is non-zero, where by convention the degree of the zero polynomial is $- \infty$. The ring of polynomials over $\F$ in $n$ 
variables is denoted by $\F[x_1, \ldots, x_n]$. 

Every polynomial $P \in \F[x_1, \ldots, x_n]$ defines naturally a function from $\F^n$ to $\F$.
While over infinite fields 
this application is one-to-one, this is not true over finite fields where different polynomials
might define the same function. For example, over the field $\F_q$ of size $q$, the polynomial $x^q -x$ is not 
the zero polynomial (it has degree $q$), but it computes the zero function. 

Numerous results are known about the properties of zero sets of polynomials. The Combinatorial 
Nullstellensatz of Alon~\cite{Alon99} is a higher dimensional extension of the well known fact that a
non-zero polynomial of degree $d$ has at most $d$ zeros. It was widely used to prove a variety of results,
among others, 
in combinatorics, graph theory and additive number theory. 

\begin{theorem}[Combinatorial Nullstellensatz]
\label{thm:cnsz}
Let $\F$ be a field, let $d_1, \ldots, d_n$ be non-negative integers, and let 
$P \in \F[x_1, \ldots, x_n]$ be a polynomial.  Suppose that $\deg(P) =\sum_{i=1}^n d_i$, and
that the coefficient of $x_1^{d_1} \cdots x_n^{d_n}$ is non-zero. Then for all subsets $S_1, \ldots, S_n$ of $\F$ with
$|S_i| > d_i$, for $i = 1, \ldots, n$, there exists $(s_1, \ldots s_n) \in S_1 \times \cdots \times S_n$ such that
$P(s_1, \ldots, s_n) \neq 0$.
\end{theorem}

The classical result of Chevalley~\cite{Ch36} and Warning~\cite{W36} 
asserts 
that if the sum of degrees of some
polynomials is less than the number of variables, than the number of their common zeros is divisible by the characteristic of 
the field.

\begin{theorem}[Chevalley-Warning Theorem]
\label{thm:cw}
Let $\F$ be a field of characteristic $p$, and let $P_1, \ldots , P_k \in \F[x_1, \ldots, x_n]$ be non-zero
polynomials. If $\sum_{i=1}^k \deg(P_i) < n$, then the number of common zeros of $P_1, \ldots , P_k$
is divisible by $p$. In particular, if the polynomials have a common zero, they also have another one.
\end{theorem}

Both of these results clearly suggest a computational problem in TFNP: Given a (set of) polynomial(s) satisfying
the respective condition of these theorems, find an element in $\F^n$ satisfying the respective conclusion.
We study here these problems 
over the two-element 
field $\F_2$ 
where both theorems have a particularly simple form, in fact they become almost the same statement.
To see that, let us recall that a {\em multilinear polynomial}  
is a polynomial of the form
$M(x_1, \ldots, x_n) = \sum_{T \subseteq \{1, \ldots , n\}} c_T x_T$, where $x_T$ stands for 
the {monomial} $\prod_{i \in T} x_i$,
and the coefficients $c_T$ are elements of $\F_2$. 
We say that a monomial $x_T$ is {\em in} $M$ if $c_T = 1$.
The degree of a multilinear polynomial $M$ is the cardinality of the largest set $T$ such that $x_T$ is in $M$.
It is well known that for every polynomial $P$ over $\F_2$, there exists
a unique multilinear polynomial $M_P(x_1, \ldots , x_n)$ such that $P$ and $M_P$ compute the same function.
We define the {\em multilinear degree} of a polynomial $P$ over $\F_2$ by
$\mdeg(P) = \deg(M_P)$. 
We call a monomial {\em maximal} if its multilinear degree is $n$. 
Clearly $\mdeg(P) \leq \deg(P)$, and $\mdeg(P) = n$
if and only if the number of maximal monomials of $P$ is odd.
Using the notion of multilinear degree, we can now state the rather simple equivalent formulations
of the above theorems over $\F_2$.

\begin{theorem}[Combinatorial Nullstellensatz over $\F_2$]
Let $P \in \F_2[x_1, \ldots, x_n]$ be a polynomial such that $\mdeg(P)=n$.
Then there exists $a \in \F_2^n$ such that
$P(a) = 1$. 
\end{theorem}

\begin{theorem}[Chevalley-Warning Theorem over $\F_2$]
Let $P \in \F_2[x_1, \ldots, x_n]$ be a polynomial such that $\mdeg(P)<n$, and let
$a \in \F_2^n$ such that
$P(a) = 0$. Then there exists $b \neq a$ such that $P(b) = 0$.
\end{theorem}

\subsection{Arithmetic circuits}\label{subsec:ckt}

An $n$-variable, $m$-output {\em arithmetic circuit} $C$ over a field $\F$ is
a vertex-labeled, acyclic directed graph whose vertices are called {\em gates}. It has $n$ {\em variable gates} of 
in-degree $0$, labeled by the variables 
$x_1, \ldots, x_n$. 
There is at most one {\em constant gate} of in-degree 0, labeled by the constant, for each non-zero field element.
The variable and constant gates are called {\em input gates}.
The other gates are of in-degree $2$, and are called {\em computational gates}. They are labeled
by $+$ or $\times$, 
the former are the {\em sum gates}, and the latter the
{\em product gates}. The number of computational gates of out-degree 0 is $m$, and they are called the {\em output gates}.

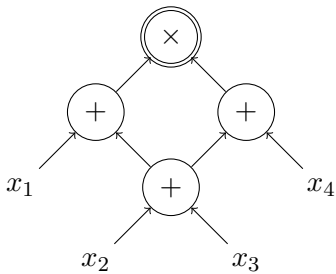
\begin{figure}[H]
	\centering
\begin{tikzpicture}
	[output/.style = {double, double distance = 1pt, circle, draw},
	gate/.style = {circle, draw},
	input/.style = {}]
	\node [output] (s1) at (0, 0) {$\times$};
	\node [gate] (s2) at (-1, -1) {$+$};
	\node [gate] (s3) at (1, -1) {$+$};
	\node [gate] (s4) at (0, -2) {$+$};
	\node [input] (x1) at (-2, -2) {$x_1$};
	\node [input] (x2) at (-1, -3) {$x_2$};
	\node [input] (x3) at (1, -3) {$x_3$};
	\node [input] (x4) at (2, -2) {$x_4$};

	\draw [->] (s2) -- (s1);
	\draw [->] (s3) -- (s1);
	\draw [->] (s4) -- (s2);
	\draw [->] (s4) -- (s3);
	\draw [->] (x1) -- (s2);
	\draw [->] (x2) -- (s4);
	\draw [->] (x3) -- (s4);
	\draw [->] (x4) -- (s3);
\end{tikzpicture}
	\caption{A $4$-variable, single-output arithmetic circuit.}
	\label{fig:circuit} 
\end{figure}
For a computational gate $g$, we distinguish its two children, 
by specifying the {\em left} and the {\em right} child.
The left child is denoted by $g_{\ell}$ and the right child by $g_r$.
We denote the set of sum gates by $G^+$, and the set of product gates by 
$G^{\times}$.
The {\em size} of $C$ is the number of its gates, and the 
{\em depth} of $C$ is the length of the longest path from an input gate to an output gate. 

The definition of an arithmetic circuit can be extended naturally to include computational gates 
of in-degree different from 2. Unary computational gates by definition 
act as the identity operator. The children of computational gates of in-degree $k >2$
are distinguished by some some distinct labeling over some set of size $k$.
It is easy to see that such an extended circuit can be simulated by a circuit with binary
computational gates,
which computes the same polynomial, and has only a 
polynomial blow-up 
in size. 
Our default circuits will be with binary computational gates, and we will mention explicitly when this is not the case.

A {\em subcircuit} of a circuit $C$ is a subgraph of $C$ which is also a circuit.
The {\em subcircuit rooted at} gate $g$
is the subgraph induced by all vertices contained on some path from the input gates to $g$, it
will be denoted by $C_g$. 
The {\em left subcircuit of} $C$, denoted by $C_{\ell}$, is the subcircuit rooted at 
the left child of the root of $C$, and the {\em right subcircuit} $C_r$ is defined similarly. 
The {\em composition} of arithmetic circuits is defined in a natural way. If $C_1$ is an $n$-variable, $m$-output circuit 
and $C_2$ is a $k$-variable, $n$-output circuit then $C_1 \circ C_2$ is the $k$-variable, $m$-output circuit
composed of $C_1$ and $C_2$ where the output 
gates of $C_1$ are identified with the variable~gates~of~$C_2$, and the identical constant gates
of the two circuits are~also~identified. Let $C_1$ and $C_2$ be $n$-variable, single-output arithmetic circuit.
The {\em disjoint sum} $C_1 \oplus C_2$ of $C_1$ and $C_2$ is the $n$-variable, single-output arithmetic circuit whose output
gate is a sum gate, its left and right subcircuits are disjoint copies of $C_1$ and $C_2$ except for the input gates
that $C_1$ and $C_2$ share. The disjoint sum naturally generalizes to more than two circuits.

Every gate $g$ in an arithmetic circuit computes an $n$-variable polynomial $P_g(x)$ in the natural way, 
which can be defined by recursion on the depth of the gate. An input gate $g$ labeled by 
$\alpha \in \{x_1, \ldots, x_n\} ~\cup ~\F$ computes $P_g =\alpha$. 
If $g \in G^+$ then $P_g = P_{g_{\ell}} + P_{g_r}$, if $g \in G^{\times}$ then $P_g = P_{g_{\ell}} P_{g_r}$.
The polynomial computed by a single-output 
arithmetic circuit $C$ is the polynomial computed by its output gate, which we will denote by $C(x).$
We define similarly by recursion the {\em circuit degree} $\cdeg(C)$ of $C$. 
If an input gate $g$ is labeled by $\alpha \in  \F$ then $\cdeg(C_g)=0$, and if it is labeled by
$\alpha \in \{x_1, \ldots, x_n\}$ then $\cdeg(C_g)=1$.
For computational gates, if $g \in G^+$ then 
$\cdeg(C_g ) = \max \{ \cdeg(C_{g_{\ell}}), \cdeg(C_{g_r})\}$, and if $g \in G^{\times}$ then 
$\cdeg(C_g ) = \cdeg(C_{g_{\ell}}) + \cdeg(C_{g_r})$.
The circuit degree can be computed in polynomial time, and we clearly have $\deg(C(x)) \leq \cdeg(C)$.


Over the base field $\F_2$, we call an element
$a \in \F_2^n$, such that $C(a) = 1$, a {\em satisfying assignment} for $C$, and an element
$a$, such that $C(a) = 0$,  a {\em zero} of $C$. For every $a \in \F_2^n$, we define the {\em Lagrange-circuit}
$L_a$ as $ C_1 \times \cdots \times C_n$, where $C_i = x_i$ if $a_i =1$, and $C_i = x_i +1$ if $a_i =0$.
Clearly $\mdeg(L_a(x)) = n$, and the only satisfying assignment for $L_a$ is $a$.

\begin{figure}[H]
	\centering
\begin{tikzpicture}
	[output/.style = {double, double distance = 1pt, circle, draw},
	gate/.style = {circle, draw},
	input/.style = {}]
	\node [output] (s1) at (0, 0) {$\times$};
	\node [gate] (s2) at (0, -1.5) {$+$};
	\node [gate] (s3) at (1, -1.5) {$+$};
	\node [input] (x1) at (-1, -1.5) {$x_1$};
	\node [input] (x2) at (-0.5, -3) {$x_2$};
	\node [input] (c1) at (0.5, -3) {$1$};
	\node [input] (x3) at (1.5, -3) {$x_3$};

	\draw [->] (s2) -- (s1);
	\draw [->] (s3) -- (s1);
	\draw [->] (x1) -- (s1);
	\draw [->] (x2) -- (s2);
	\draw [->] (x3) -- (s3);
	\draw [->] (c1) -- (s2);
	\draw [->] (c1) -- (s3);
\end{tikzpicture}
	\caption{Lagrange-circuit $L_{100}$.}
	\label{fig:lagrange} 
\end{figure}
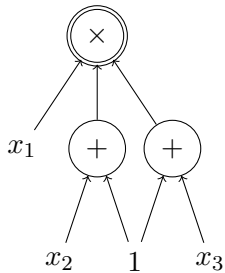

\section{Parse subcircuits}
\label{sec:parse_sub}

\suppress{
\alex{It is a good idea to add dramatic effect, but imho, it should be done a bit differently.
Start with formulae.  For them, it is easy to see how monomials are formed: add left or right semaphore to each accessible addition gate.
If we go from formulae to general circuits, we lose this simplicity since for different paths we might need to go in different directions.  But look, over characteristic 2, we get the same simplicity as for formulae also for general circuits.}
\miklos{I added half a sentence in that direction ``This is in fact the case ...", 
but otherwise I prefer to leave it this way. Actually the computation
of monomials in circuits is an interesting subject, and I feel that by starting with formulae this is not even hinted to. But feel free to change it to your taste}
\alex{``This is in fact the case'' is not really good here, since for formulae we need that there is no assignment for non-accessible sum-gates.}
}

We would like to understand how monomials are computed by a single-output
arithmetic circuit $C$.
If $g$ is a sum gate, then the set of monomials computed by $C_g$  is a subset of the union of the set of monomials
computed by $C_{g_{\ell}}$ and by $C_{g_r}$.
If $g$ is a multiplication gate, then every monomial computed by $C_g$ is the product of a monomial 
computed by $C_{g_{\ell}}$ and 
a monomial computed by $C_{g_r}$.  A marking of the gates in $G^+$ from the set $\{\ell, r\}$
therefore computes naturally a monomial of $C(x)$. At first sight it seems that 
by considering markings restricted to the sum gates effectively participating in the computing of the monomial,
we could compute all of them.
This is in fact the case 
when the fanout of every sum gate is one, but this is not true in general circuits since the sum gates can be used
several times in the computation of a monomial with possibly inconsistent markings. However, as we show it below, this is
essentially true over fields of characteristic 2, where it is sufficient to consider only consistent markings.
By doing that, we have to be careful about two things: when computing a 
monomial by some marking, we shouldn't mark those sum gates which don't participate in its computation. Indeed, by 
considering the two possible markings also for irrelevant gates, we would  assure that the monomial 
is necessarily computed an even number of times, making the whole process false. 
On the other hand, we should mark all the sum gates necessary for the 
computation of the monomial. 
We make all this precise by the notion of closed marking and parse subcircuit.


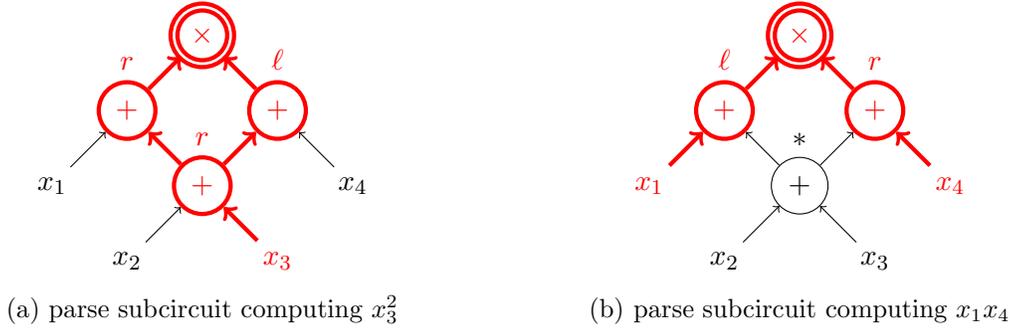
\begin{figure}[H]
	\centering
	\begin{subfigure}[b]{0.45\textwidth}
	\centering
\begin{tikzpicture}
	[output/.style = {double, double distance = 1pt, circle, draw},
	gate/.style = {circle, draw},
	input/.style = {},
	parse/.style = {color = red, ultra thick}]
	\node [output, parse] (s1) at (0, 0) {$\times$};
	\node [gate, parse] (s2) at (-1, -1) [label = above:{\color{red}$r$}] {$+$};
	\node [gate, parse] (s3) at (1, -1) [label = above:{\color{red}$\ell$}] {$+$};
	\node [gate, parse] (s4) at (0, -2) [label = above:{\color{red}$r$}] {$+$};
	\node [input] (x1) at (-2, -2) {$x_1$};
	\node [input] (x2) at (-1, -3) {$x_2$};
	\node [input, parse] (x3) at (1, -3) {$x_3$};
	\node [input] (x4) at (2, -2) {$x_4$};

	\draw [->, parse] (s2) -- (s1);
	\draw [->, parse] (s3) -- (s1);
	\draw [->, parse] (s4) -- (s2);
	\draw [->, parse] (s4) -- (s3);
	\draw [->] (x1) -- (s2);
	\draw [->] (x2) -- (s4);
	\draw [->, parse] (x3) -- (s4);
	\draw [->] (x4) -- (s3);
\end{tikzpicture}
		\subcaption{parse subcircuit computing $x_3^2$}
	\end{subfigure}
	\quad
	\begin{subfigure}[b]{0.45\textwidth}
	\centering
\begin{tikzpicture}
	[output/.style = {double, double distance = 1pt, circle, draw},
	gate/.style = {circle, draw},
	input/.style = {},
	parse/.style = {color = red, ultra thick}]
	\node [output, parse] (s1) at (0, 0) {$\times$};
	\node [gate, parse] (s2) at (-1, -1) [label = above:{\color{red}$\ell$}] {$+$};
	\node [gate, parse] (s3) at (1, -1) [label = above:{\color{red}$r$}] {$+$};
	\node [gate] (s4) at (0, -2) [label = above:$*$] {$+$};
	\node [input, parse] (x1) at (-2, -2) {$x_1$};
	\node [input] (x2) at (-1, -3) {$x_2$};
	\node [input] (x3) at (1, -3) {$x_3$};
	\node [input, parse] (x4) at (2, -2) {$x_4$};

	\draw [->, parse] (s2) -- (s1);
	\draw [->, parse] (s3) -- (s1);
	\draw [->] (s4) -- (s2);
	\draw [->] (s4) -- (s3);
	\draw [->, parse] (x1) -- (s2);
	\draw [->] (x2) -- (s4);
	\draw [->] (x3) -- (s4);
	\draw [->, parse] (x4) -- (s3);
\end{tikzpicture}
		\subcaption{parse subcircuit computing $x_1 x_4$}
	\end{subfigure}
	\caption{Two parse subcircuits for 
	Figure~\ref{fig:circuit}, note that 
	the second one doesn't access all sum gates.}
\label{fig:psc}
\end{figure}

Let $C$ be a single-output arithmetic circuit. A {\em marking} of $C$ is a partial function $S\colon G^+ \rightarrow \{\ell, r\}$,
from the sum gates of $C$ to the marks $\{\ell, r\}$. 
We can equivalently specify a marking by a total function $S^*\colon G^+ \rightarrow \{\ell, r, *\}$ where $S^*(g) = *$ if and only
if $S(g)$ is undefined. 
We denote by $\Dom(S)$ the domain of $S$.
For the output gate of $C$, let $S_{\ell}$ be the restriction of $S$
to the sum gates in $C_{\ell}$ and let $S_{r}$ be the restriction of $S$
to the sum gates in  $C_r$. 
We define $G_S = (V_S, E_S)$, the {\em accessibility graph } of $S$ 
by induction on the depth of $C$.
If $C$ is a single vertex then $V_S$ consists of this vertex, and $E_S = \emptyset$.
Otherwise, 
if the output gate is a product gate, then 
$V_S$ consists of the output gate of $C$ added to $V_{S_{\ell}} \cup V_{S_r}$, and
$E_S$ consists of the two edges from the two children of the output gate to the output gate, added to $E_{S_{\ell}} \cup E_{S_r}$.
If the output gate of $C$ is a sum gate with mark $\ell$ then 
$V_S$ consist of the output gate of $C$ added to $V_{S_{\ell}}$, and
$E_S$ consists of the  edge from the left child of the output gate to the output gate, added to $E_{S_{\ell}}$.
The definition in the case when the mark of the output gate is $r$ is analogous. If the output gate of $C$ doesn't have a mark then 
the accessibility graph is just this single node.

We say that a marking $S$ is {\em closed} if $\Dom(S) = V_S \cap G^+$, that is if the accessible sum gates of 
$C$ are exactly those where $S$ is defined.
If $S$ is 
closed then the accessibility graph $G_S$, with the vertex labels inherited from $C$, is in fact 
a subcircuit of $C$. The inclusion $\Dom(S) \subseteq V_S \cap G^+$ ensures that 
the only node of out-degree 0 in $G_S$ is the output gate of $C$,
and the inclusion $V_S \cap G^+ \subseteq \Dom(S)$ ensures that the
leaves of $G_S$ are leaves in $C$. We call this subcircuit the {\em parse subcircuit} induced by $S$,
and denote it by $C_S$.
The set of parse subcircuits of $C$ will be denoted by $\cS(C)$.
Observe that a parse subcircuit has binary product gates but unary sum gates which act as the identity operator.
The polynomial $C_S(x)$ {\em computed by} the parse subcircuit $C_S$ 
is therefore a monomial, which we denote by $m_S(x)$.
We say that a parse subcircuit $C_S$ 
is {\em maximal}
if the multilinear degree of $m_S(x)$ is $n$, that is $m_S(x)= x_1 \cdots x_n$.
We say that two parse subcircuits $C_S$ and $C_{S'}$ 
are {\em consistent} if for every $g \in \Dom(S) \cap \Dom(S')$, we have $S(g) = S'(g)$.

Clearly, the mapping from closed markings to induced parse subcircuits is a 
bijection. 
Therefore, to ease notation, we will often call the
closed marking $S$ itself the parse subcircuit, and we will speak about the gates, subcircuits 
and other circuit related notions of $S$, instead of $C_S$. The notation used for the monomial
computed by a 
parse 
subcircuit is already consistent with this convention.


\begin{proposition}
\label{prop:char2}
Let $C$ be a single-output arithmetic circuit over a field $\F$ of characteristic $2$. Then
$$C(x) =\sum_{S \in \cS(C)} m_S(x).$$ 
\end{proposition}
\begin{proof}
We prove by induction on the depth of the circuit. If $C$ consists of a single gate, the statement is obvious.

Otherwise, 
the parse subcircuits of $\cS(C_{\ell})$ (respectively $\cS(C_r)$) are exactly the parse subcircuits of $\cS(C)$
restricted to the sum gates of $C_{\ell}$ (respectively $C_r$). 
When the output gate of $C$
is a sum gate then conversely, $\cS(C)$ can be obtained from $ \cS(C_{\ell}) \cup \cS(C_r)$ by
extending the markings in the latter set with the appropriate mark for the root of $C$. 
Therefore, using the definitions of $C(x)$ and $m_S(x)$, we get
\begin{align*}
C(x) & = {C_{\ell}}(x) + {C_{r}}(x) \\
& =  \sum_{S \in \cS(C_{\ell})} m_S(x) +  \sum_{S \in \cS(C_{r})} m_S(x) \\
& =  \sum_{S \in \cS(C), ~S(\mathrm{root})= \ell} m_{{S_{\ell}}}(x)  + \sum_{S \in \cS(C), ~S(\mathrm{root})= r} m_{{S_{r}}}(x) \\
& = \sum_{S \in \cS(C)} m_S(x),
\end{align*}
where the second equality comes from the inductive hypothesis.

\begin{figure}[H]
	\centering
	\begin{subfigure}[b]{0.45\textwidth}
		\centering
\begin{tikzpicture}
	[output/.style = {double, double distance = 1pt, circle, draw},
	gate/.style = {circle, draw},
	input/.style = {}]
	\node [output] (s1) at (0, 0) {$\times$};
	\node [gate] (s2) at (-1, -1) [label = left:{\color{blue}$U$}] {$+$};
	\node [gate] (s3) at (1, -1) [label = right:{\color{green}$W$}] {$+$};
	\node [gate] (s4) at (0, -2) [label = above:$g$] {$+$};
	\node [input] (x1) at (-2, -2) {$x_1$};
	\node [input] (x2) at (-1, -3) {$x_2$};
	\node [input] (x3) at (1, -3) {$x_3$};
	\node [input] (x4) at (2, -2) {$x_4$};

	\draw [->] (s2) -- (s1);
	\draw [->] (s3) -- (s1);
	\draw [->, ultra thick, blue, dotted] (s4) -- (s2);
	\draw [->, ultra thick, green, dashed] (s4) -- (s3);
	\draw [->] (x1) -- (s2);
	\draw [->, ultra thick, blue, dotted] (x2) -- (s4);
	\draw [->, ultra thick, green, dashed] (x3) -- (s4);
	\draw [->] (x4) -- (s3);
\end{tikzpicture}
		\subcaption{inconsistent $U, W$}
	\end{subfigure}
	\quad
	\begin{subfigure}[b]{0.45\textwidth}
		\centering
\begin{tikzpicture}
	[output/.style = {double, double distance = 1pt, circle, draw},
	gate/.style = {circle, draw},
	input/.style = {}]
	\node [output] (s1) at (0, 0) {$\times$};
	\node [gate] (s2) at (-1, -1) [label = left:{\color{blue}$U'$}] {$+$};
	\node [gate] (s3) at (1, -1) [label = right:{\color{green}$W'$}] {$+$};
	\node [gate] (s4) at (0, -2) [label = above:$g$] {$+$};
	\node [input] (x1) at (-2, -2) {$x_1$};
	\node [input] (x2) at (-1, -3) {$x_2$};
	\node [input] (x3) at (1, -3) {$x_3$};
	\node [input] (x4) at (2, -2) {$x_4$};

	\draw [->] (s2) -- (s1);
	\draw [->] (s3) -- (s1);
	\draw [->, ultra thick, dotted, blue] (s4) -- (s2);
	\draw [->, ultra thick, dashed, green] (s4) -- (s3);
	\draw [->] (x1) -- (s2);
	\draw [->, ultra thick, dashed, green] (x2) -- (s4);
	\draw [->, ultra thick, dotted, blue] (x3) -- (s4);
	\draw [->] (x4) -- (s3);
\end{tikzpicture}
		\subcaption{inconsistent $U', W'$}
	\end{subfigure}
	\caption{The involutive pair $(U,W) \leftrightarrow (U',W')$ 
	in the proof of Proposition~\ref{prop:char2}
	with $m_U m_W = x_2 x_3 = m_{U'} m_{W'}$ contributes zero to $C(x)$.}
\end{figure}
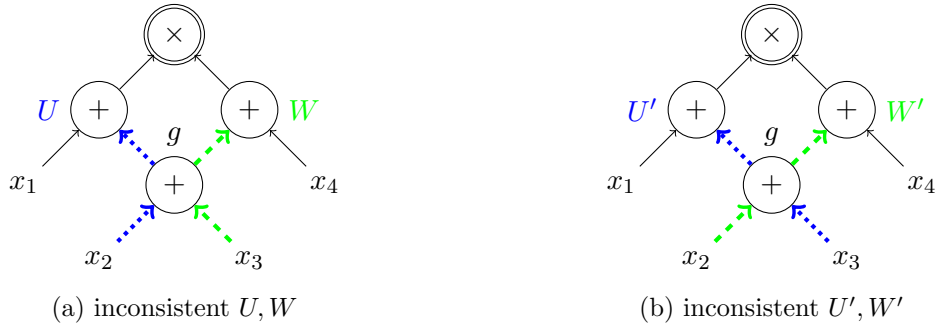

When the output gate of $C$ is a product
gate, the situation is more complicated. The parse subcircuits $S_{\ell}$ and $S_r$ are always consistent for 
$S \in \cS(C)$, but an arbitrary parse subcircuit $U \in \cS(C_{\ell})$ is not necessarily consistent with an arbitrary
parse subcircuit $W \in \cS(C_{r})$. Therefore the crux of the induction step is to show that the contribution
of $m_U(x)m_W(x)$ to $C(x)$ is zero when we sum over all inconsistent $U$ and $W$. Indeed, we claim that
$$
\sum_{ (U,W)  \in \cS(C_{\ell}) \times \cS(C_{r}), ~U,W {\rm  inconsistent}} m_U(x)m_W(x) = 0 .
$$

To prove this, we define an involution 
$(U,W) \leftrightarrow (U',W')$ 
over inconsistent pairs in $\cS(C_{\ell}) \times \cS(C_{r})$
such that $m_U(x)m_W(x) + m_{U'}(x)m_{W'}(x) = 0$.
For this let us fix some topological ordering of the gates in $C$ with respect to the edges of the circuit,
and let $g$ be the first sum gate 
in this ordering where $U$ and $W$ have different marks, 
say $U(g) = \ell$ and $W(g) = r$.
Let the restriction of $U$ to the sum gates of $C_g$ be $T_0$ and let the
restriction of $W$ to the sum gates of $C_g$ be $T_1$. Both
$T_0$ and $T_1$ are parse subcircuits in $C_g$, which are inconsistent only at $g$. Also,
for some monomials $m_0(x)$ and $m_1(x)$, we have $m_U(x) = m_0(x) m_{T_0}(x)$ 
and
$m_W(x) = m_1(x) m_{T_1}(x)$. The parse subcircuit $U'$ is obtained from $U$ by exchanging inside $C_g$ the 
parse subcircuit $T_0$ for the parse subcircuit $T_1$, that is $U' = (U \setminus T_0) \cup T_1$.
The parse subcircuit $W'$ is similarly defined from $W$ with the roles of $T_0$ and $T_1$ reversed.
It follows from the choice of $g$ that
$U'$ and $W'$ are parse subcircuits respectively in $ \cS(C_{\ell})$ and $\cS(C_{r})$ such that 
the first inconsistency between 
them in the topological order is at $g$. 
Therefore starting the same process
with $(U',W')$ we obtain $(U,W)$, and thus the mapping is indeed an involution.
Since $m_{U'}(x) = m_0(x) m_{T_1}(x)$ and
$m_{W'}(x) = m_1(x) m_{T_0}(x)$, we can conclude that $m_U(x)m_W(x) + m_{U'}(x)m_{W'}(x) = 0$.

We can now complete the induction step for product gates by observing the equalities

\begin{align*}
C(x) & =  {C_{\ell}}(x)  \times {C_{r}}(x) \\
& =   \left(\sum_{U \in \cS(C_{\ell})} m_U(x)\right) \times  \left( \sum_{W \in \cS(C_{r})} m_W(x)\right)  \\
& =  \sum_{ (U,W)  \in \cS(C_{\ell}) \times \cS(C_{r}), ~U,W {\rm  consistent}} m_U(x)m_W(x) \\
& =  \sum_{S \in \cS(C)} m_{S_{\ell} }(x) m_{S_r}(x) \\
& =   \sum_{S \in \cS(C)} m_S(x).
\end{align*}

\end{proof}

Though it is not directly related to the main result of the paper, we prove here, essentially as a corollary of the previous
proposition, that deciding if the polynomial computed by a circuit over the two elements field 
has maximal multilinear degree is $\oplus \Pe$-complete. 
Note that by the Chevalley-Warning theorem, the multilinear degree of a circuit is 
maximal if and only if it has odd number of satisfying assignments, and via this 
correspondence 
Proposition~\ref{prop:parPcomplete} can also be proved by using the number of 1's 
to build a balanced relation. The point of our proof of 
Proposition~\ref{prop:parPcomplete} 
is to show this without referring to the Chevalley-Warning theorem, and therefore 
illustrate the use of maximal parse subcircuits.  

\begin{proposition}
\label{prop:parPcomplete}
Let $C$ be an $n$-variable, single-output arithmetic circuit over the field $\F_2$. The problem of
deciding if $\mdeg(C(x)) = n$ is $\oplus \Pe$-complete.
\end{proposition}
\begin{proof}

For the easiness part, we can define a balanced relation $R(C, S)$ where $S \in \mathcal{S}(C)$, 
which equals $1$ if and only if $S$ is a maximal parse subcircuit.
By Proposition~\ref{prop:char2}, we know that the polynomial computed by the circuit $C$ is 
the sum of all the monomials computed by the parse subcircuits.
Among all the parse subcircuits, only the monomials computed
by maximal parse subcircuits have degree $n$. 
Thus $\textrm{mdeg}(C(x)) = n$ if and only if there is an odd number of maximal parse
subcircuits. 


For the hardness part, we will reduce the well known $\oplus \Pe$-complete problem $\oplus 3$-SAT~\cite{Valiant05}
to the maximality of $\textrm{mdeg}(C(x))$.
Let  $\phi = \{F_1, F_2, \dots, F_m\}$ be an instance of $3$-SAT, where the clause $F_i$ is the conjunction
of three literals 
belonging to $\{x_1, \overline{x_1}, \dots, x_n, \overline{x_n}\}$.
The reduction maps $\phi$ to an $m$-variable, single-output and depth-$3$ arithmetic circuit $C$ defined as follows.
The output gate at level 0 is a product gate. It has $n$ children
$\alpha_1, \dots, \alpha_n$, all plus gates, which compose the first level of the circuit. At level 2, for all $1 \leq j \leq n$,
the gate $\alpha_j$ has two children $x_j$ and $\overline{x_j}$, which are product gates.
The gate $x_j$ is the left child of $\alpha_j$, and $\overline{x_j}$ is its right child.
Finally at level 3 are the $m$ variable gates $F_1, \ldots , F_m$, such that $F_i$ is a child of
$y \in \{x_1, \overline{x_1}, \dots, x_n, \overline{x_n}\}$ if $y \in F_i$ in $\phi$.
The following is an illustration of the circuit which is the image of the formula
$(x_1 \vee x_2 \vee \overline{x_3}) \wedge (\overline{x_2} \vee x_3) \wedge (x_3 \vee \overline{x_1})$ by the reduction.

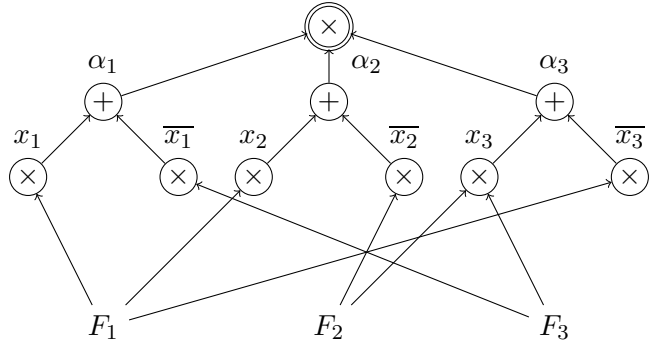
\begin{figure}[H]
	\centering
\begin{tikzpicture}
	[output/.style = {double, double distance = 1pt, circle, draw},
	gate/.style = {circle, draw},
	small/.style = {minimum size = 2pt, inner sep = 2pt},
	ultra small/.style = {minimum size = 1pt, inner sep = 1pt},
	input/.style = {}]
	\node [output, small] (s0) at (0, 0) {$\times$};
	\node [gate, ultra small] (s1) at (-3, -1)  [label = {$\alpha_1$}] {$+$};
	\node [gate, ultra small] (s2) at (0, -1) [label = {[xshift=0.5cm]$\alpha_2$}] {$+$};
	\node [gate, ultra small] (s3) at (3, -1) [label = {$\alpha_3$}] {$+$};
	\node [gate, ultra small] (s4) at (-4, -2) [label = {$x_1$}] {$\times$};
	\node [gate, ultra small] (s5) at (-2, -2) [label = {$\overline{x_1}$}] {$\times$};
	\node [gate, ultra small] (s6) at (-1, -2) [label = {$x_2$}] {$\times$};
	\node [gate, ultra small] (s7) at (1, -2) [label = {$\overline{x_2}$}] {$\times$};
	\node [gate, ultra small] (s8) at (2, -2) [label = {$x_3$}] {$\times$};
	\node [gate, ultra small] (s9) at (4, -2) [label = {$\overline{x_3}$}] {$\times$};
	\node [input] (x1) at (-3, -4) {$F_1$};
	\node [input] (x2) at (0, -4) {$F_2$};
	\node [input] (x3) at (3, -4) {$F_3$};

	\draw [->] (s1) -- (s0);
	\draw [->] (s2) -- (s0);
	\draw [->] (s3) -- (s0);

	\draw [->] (s4) -- (s1);
	\draw [->] (s5) -- (s1);
	\draw [->] (s6) -- (s2);
	\draw [->] (s7) -- (s2);
	\draw [->] (s8) -- (s3);
	\draw [->] (s9) -- (s3);

	\draw [->] (x1) -- (s4);
	\draw [->] (x1) -- (s6);
	\draw [->] (x1) -- (s9);

	\draw [->] (x2) -- (s7);
	\draw [->] (x2) -- (s8);

	\draw [->] (x3) -- (s5);
	\draw [->] (x3) -- (s8);
\end{tikzpicture}
	\caption{Image of $(x_1 \vee x_2 \vee \overline{x_3}) \wedge (\overline{x_2} \vee x_3) \wedge (x_3 \vee \overline{x_1})$
	by the reduction.}
	\label{fig:3sat} 
\end{figure}

We give a one-to-one mapping $S$ from the assignments of $\phi$ to the parse subcircuits of $\mathcal{S}(C)$.
Since all plus gates of $C$ are reachable from the output gate, a parse subcircuit of $C$ is an $\{\ell, r\}$-marking
of the gates $\alpha_1, \dots, \alpha_n$. The parse subcircuits are therefore naturally identified with the elements of 
$\{\ell, r\}^n$.
For an assignment $x \in \{0,1\}^n,$ the map $S$ is defined as
\[ S(x)_i = \left\{ \begin{array}{ll}
	\ell & \textrm{if } x_i = 1\\
	r & \textrm{if } x_i = 0.\\
\end{array} \right. \]

To finish the proof we show that $x$ is a satisfying assignment if and only if $S(x)$ is a maximal parse subcircuit. To see that, observe
that $x$ is a satisfying assignment if and only if each $F_i$ in $\phi$ contains a true literal. By the definition of $S$, the clause
$F_i$ contains a true literal exactly when the variable $F_i$ of $C$ is in the parse subcircuit $C_{S(x)}$.
Since $C_{S(x)}$ is maximal if and only if $F_i$ is in the parse subcircuit $C_{S(x)}$ for all $i$, this concludes the proof.
\end{proof}

\section{PPA-circuits}
\label{sec:ppa-circuit}

Given an arbitrary circuit $C$ and a satisfying assignment, asking
for another satisfying assignment would be an $\NP$-hard problem. We want to restrict
the form of the circuit $C$ in a way which takes into consideration the structure of problems in $\PPA$.
For this, we use repeatedly a $2n$-variable, single-output arithmetic circuit $I$.
The circuit $I$ is of depth 2, its output gate is a product gate with $n$ children, all sum gates. Every sum gate has 3 children,
the left child of the $i$th gate is the variable gate $x_i$, its center child is the variable gate $y_i$, 
and its right child is the constant gate 1. 
For an $n$-variable, $n$-output circuit $C$, we define $I \diamond C$, the {\em diamond composition} of $I$ with $C$,
as the $n$-variable, single-output circuit  composed from a circuit $I$ at the top and $C$ below.
More precisely, the variable gates  of $I \diamond C$ 
labeled by $x_1, \ldots, x_n$ are also 
the first $n$ variables of $I$, and the variable gates $y_1, \ldots, y_n$ of $I$ are identified with the output gates of $C$. 
If $C$ has also a constant gate 1, it is identified with the constant gate 1 of $I$.

 \begin{figure}[H]
	\centering
	 \begin{subfigure}[b]{0.45\textwidth}
		 \centering
\begin{tikzpicture}
	[output/.style = {double, double distance = 1pt, circle, draw},
	gate/.style = {circle, draw},
	input/.style = {}]
	\node [output] (s1) at (0, 0) {$\times$};
	\node [gate] (s2) at (-1, -1) {$+$};
	\node [input] (s3) at (0, -1) {$\cdots$};
	\node [gate] (s4) at (1, -1) {$+$};
	\node [input] (x1) at (-3, -3) {$x_1$};
	\node [input] (x2) at (-2, -3) {$\cdots$};
	\node [input] (x3) at (-1, -3) {$x_n$};
	\node [input] (x4) at (0, -3) {$y_1$};
	\node [input] (x5) at (1, -3) {$\cdots$};
	\node [input] (x6) at (2, -3) {$y_n$};
	\node [input] (x7) at (3, -3) {$1$};

	\draw [->] (s2) -- (s1);
	\draw [->] (s4) -- (s1);
	\draw [->] (x1) -- (s2);
	\draw [->] (x4) -- (s2);
	\draw [->] (x7) -- (s2);
	\draw [->] (x3) -- (s4);
	\draw [->] (x6) -- (s4);
	\draw [->] (x7) -- (s4);
\end{tikzpicture}
	\caption{The arithmetic circuit $I$.}
	 \end{subfigure}
	 \quad
	 \begin{subfigure}[b]{0.45\textwidth}
		 \centering
	\begin{tikzpicture}
		[gate/.style = {circle, draw, fill = white},
		small/.style = {minimum size = 1pt, inner sep = 1pt}]
		\draw (3, 1.5) -- (3, 3) -- (0, 3) -- (0, 1.5) -- (3, 1.5);

		\node at (1.5, 2.25) {$C$};

		\node [] (x1) at (-2.5, 0) {$x_1$};
		\node [] (x2) at (-1.5, 0) {$\cdots$};
		\node [] (x3) at (-0.5, 0) {$x_n$};

		\node (const1) at (3.5, 3) {$1$};
		\node [gate, small] (i1) at (-0.5, 5) {$+$};
		\node (i2) at (0.5, 5) {$\cdots$};
		\node [gate, small] (i3) at (1.5, 5) {$+$};

		\node [gate, small] (t) at (0.5, 6) {$\times$};

		\draw [->] (i1) -- (t);
		\draw [->] (i3) -- (t);

		\node [gate] (d1) at (0.5, 3) {};
		\node [gate, draw = white] (d2) at (1.5, 3) {$\cdots$};
		\node [gate] (d3) at (2.5, 3) {};

		\draw [->] (x1) -- (i1);
		\draw [->] (d1) -- (i1);
		\draw [->] (const1) -- (i1);

		\draw [->] (x3) to [bend left = 25] (i3);
		\draw [->] (d3) -- (i3);
		\draw [->] (const1) -- (i3);

		\node [gate] (x1') at (0.5, 1.5) {};
		\node [gate, draw = white] (x2') at (1.5, 1.5) {$\cdots$};
		\node [gate] (x3') at (2.5, 1.5) {};

		\draw [->] (x1) -- (x1');
		\draw [->] (x3) -- (x3');

	\end{tikzpicture}
	\caption{The compound arithmetic circuit $I \diamond C$.}
	 \end{subfigure}
	\caption{The arithmetic circuits $I$ and $I \diamond C$.} 
\end{figure}
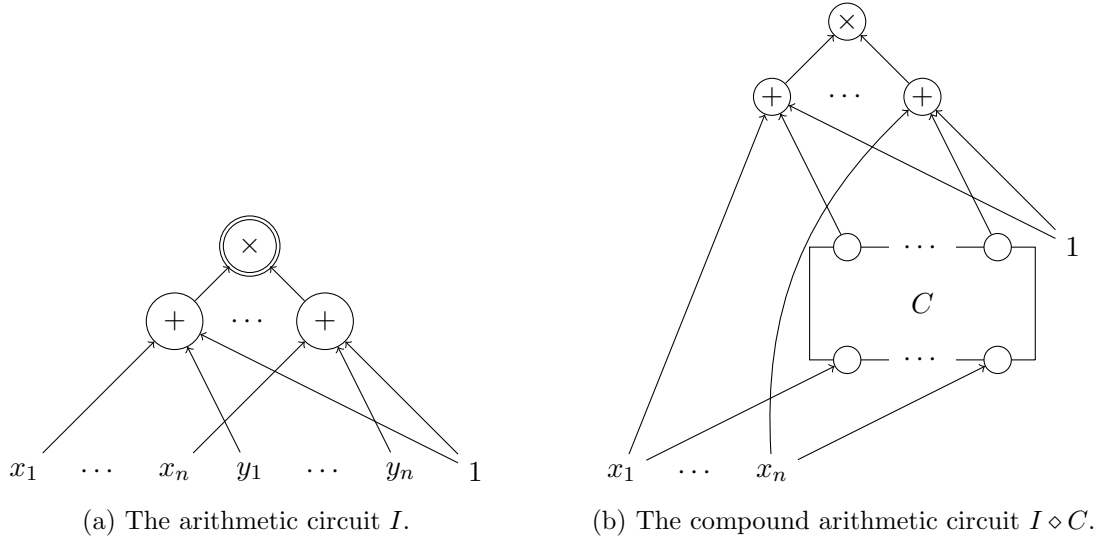

The polynomial computed by the circuit $I$ is 
$I(x_1, \ldots, x_n, y_1, \ldots, y_n) = \prod_{i = 1}^n (x_i + y_i + 1)$. It is easy to check that $I(x,y)$ is 1
if and only if the two $n$-bit strings $x_1, \ldots, x_n$ and $y_1, \ldots, y_n$ are equal.  Therefore $I \diamond C(x) = 1$
if and only if $C(x) = x$.

Given two $n$-variable,
$n$-output arithmetic circuits $D$ and $F$, 
we consider the set of six $n$-variable, single-output circuits 
$$
\ctabc_{D,F} = \{I_1 \diamond D_1 \circ F_1,  I_2 \diamond F_2 \circ D_2, I_3 \diamond D_3 \circ D_4 , I_4 \diamond D_5, I_5 \diamond F_3 \circ F_4, 
I_6 \diamond F_5\},
$$
where $I_1, \ldots , I_6$ are copies of $I$; $D_1, \ldots, D_5$ are copies of $D$; $F_1, \ldots, F_5$ are copies of $F$,
and the six circuits share the same input gates. 
The {\em $\PPA$-composition} 
of $D$ and $F$ is the $n$-variable, single-output circuit $C_{D,F}$
is the disjoint sum of the six circuits in $\ctabc_{D,F}$. 
We call the circuits in $\ctabc_{D,F}$ the 
{\em components} of $C_{D,F}$.
The polynomial computed by $C_{D,F}$ is 
$$
C_{D,F}(x) = I(x, D(F(x))) + I(x, F(D(x))) +  I(x, D(D(x))) +  I(x, D(x))) +  I(x, F(F(x))) +  I(x, F(x))).
$$
\begin{figure}[H]
	\centering
\begin{tikzpicture}[scale = 0.5,
	output/.style = {double, double distance = 1pt, circle, draw},
	gate/.style = {circle, draw},
	input/.style = {},
	small/.style = {minimum size = 2pt, inner sep = 2pt},
	ultra small/.style = {minimum size = 1pt, inner sep = 1pt}]
	\node [output, small] (s) at (12.5, 1) {$+$};

	\begin{scope}[shift = {(0, 0.75)}]
		\node (x1) at (10, -6) {$x_1$};
		\node at (12.5, -6) {$\cdots$};
		\node (xn) at (15, -6) {$x_n$};
	\end{scope}

		\draw (-2, -2) -- (-2, 0) -- (2, 0) -- (2, -2) -- (-2, -2);
		\node at (0, -1) {$I_1 \diamond D_1 \circ F_1$};
		\node [gate, small, fill = white] (c) at (0, 0) {};
		\draw [->] (c) -- (s);
		\node [gate, ultra small, draw = white, fill = white] at (0, -2) {$\cdots$};
		\node [gate, small, fill = white] (d1) at (-1.5, -2) {};
		\node [gate, small, fill = white] (dn) at (1.5, -2) {};
		\draw [->] (x1) -- (d1);
		\draw [->] (xn) -- (dn);

	\begin{scope}[shift = {(5, 0)}]
		\draw (-2, -2) -- (-2, 0) -- (2, 0) -- (2, -2) -- (-2, -2);
		\node at (0, -1) {$I_2 \diamond F_2 \circ D_2$};
		\node [gate, small, fill = white] (c) at (0, 0) {};
		\draw [->] (c) -- (s);
		\node [gate, ultra small, draw = white, fill = white] at (0, -2) {$\cdots$};
		\node [gate, small, fill = white] (d1) at (-1.5, -2) {};
		\node [gate, small, fill = white] (dn) at (1.5, -2) {};
		\draw [->] (x1) -- (d1);
		\draw [->] (xn) -- (dn);
	\end{scope}

	\begin{scope}[shift = {(10, 0)}]
		\draw (-2, -2) -- (-2, 0) -- (2, 0) -- (2, -2) -- (-2, -2);
		\node at (0, -1) {$I_3 \diamond D_3 \circ D_4$};
		\node [gate, small, fill = white] (c) at (0, 0) {};
		\draw [->] (c) -- (s);
		\node [gate, ultra small, draw = white, fill = white] at (0, -2) {$\cdots$};
		\node [gate, small, fill = white] (d1) at (-1.5, -2) {};
		\node [gate, small, fill = white] (dn) at (1.5, -2) {};
		\draw [->] (x1) -- (d1);
		\draw [->] (xn) -- (dn);
	\end{scope}

	\begin{scope}[shift = {(15, 0)}]
		\draw (-2, -2) -- (-2, 0) -- (2, 0) -- (2, -2) -- (-2, -2);
		\node at (0, -1) {$I_4 \diamond D_5$};
		\node [gate, small, fill = white] (c) at (0, 0) {};
		\draw [->] (c) -- (s);
		\node [gate, ultra small, draw = white, fill = white] at (0, -2) {$\cdots$};
		\node [gate, small, fill = white] (d1) at (-1.5, -2) {};
		\node [gate, small, fill = white] (dn) at (1.5, -2) {};
		\draw [->] (x1) -- (d1);
		\draw [->] (xn) -- (dn);
	\end{scope}

	\begin{scope}[shift = {(20, 0)}]
		\draw (-2, -2) -- (-2, 0) -- (2, 0) -- (2, -2) -- (-2, -2);
		\node at (0, -1) {$I_5 \diamond F_3 \circ F_4$};
		\node [gate, small, fill = white] (c) at (0, 0) {};
		\draw [->] (c) -- (s);
		\node [gate, ultra small, draw = white, fill = white] at (0, -2) {$\cdots$};
		\node [gate, small, fill = white] (d1) at (-1.5, -2) {};
		\node [gate, small, fill = white] (dn) at (1.5, -2) {};
		\draw [->] (x1) -- (d1);
		\draw [->] (xn) -- (dn);
	\end{scope}

	\begin{scope}[shift = {(25, 0)}]
		\draw (-2, -2) -- (-2, 0) -- (2, 0) -- (2, -2) -- (-2, -2);
		\node at (0, -1) {$I_6 \diamond F_5$};
		\node [gate, small, fill = white] (c) at (0, 0) {};
		\draw [->] (c) -- (s);
		\node [gate, ultra small, draw = white, fill = white] at (0, -2) {$\cdots$};
		\node [gate, small, fill = white] (d1) at (-1.5, -2) {};
		\node [gate, small, fill = white] (dn) at (1.5, -2) {};
		\draw [->] (x1) -- (d1);
		\draw [->] (xn) -- (dn);
	\end{scope}
	
\end{tikzpicture}
\caption{The circuit $C_{D, F}$, the $\PPA$-composition of the circuits $D$ and $F$.}
\label{fig:PPAcomp}
\end{figure}
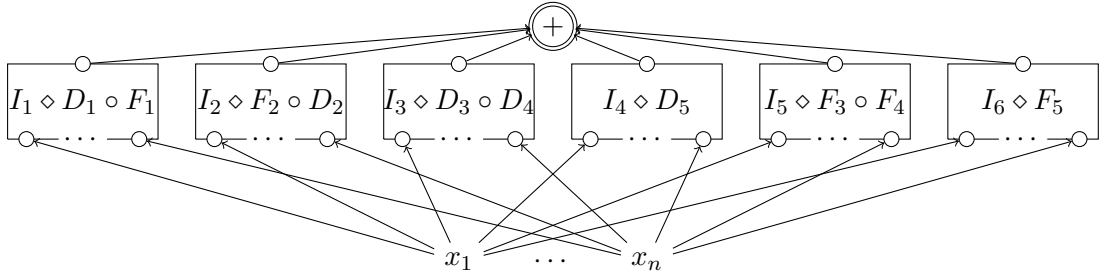

The main structural property of a PPA-composition $C$ is 
that it computes a polynomial whose multilinear degree is 
less than $n$. Moreover, a witness for that can be computed in polynomial time.
By Proposition~\ref{prop:char2}, the multilinear degree of $C(x)$ is determined by the parity of its maximal parse subcircuits, 
$\mdeg(C(x)) =n$ if and only if the parity of the maximal parse subcircuits 
is odd.
Thus, 
the multilinear degree of $C(x)$ can be certified by a special type of 
syntactically defined 
matching over its maximal parse subcircuits.
Formally, a {\em  matching for maximal parse subcircuits in} $C$ is a polynomial time Turing machine $\mu$ which defines
a matching over the maximal parse subcircuits of $C$ as follows:
$S$ and $S'$ are {\em matched} if $\mu(C,S) = S'$ and $\mu(C,S') = S$.
If $\mu$ defines a
perfect matching between the maximal parse subcircuits, then $\mdeg(C(x)) < n$.
If $\mu$ defines a
perfect matching {\em outside} some maximal parse subcircuit $T$, meaning that $T$ is the only maximal parse subcircuit
without a matching pair in $\mu$,
then $\mdeg(C(x)) = n.$

All the above statements hold also for circuits which are the direct sum of a PPA-composition and another circuit which certifiably
has no maximal parse subcircuit. This is obviously the case of circuits which compute polynomials of degree less than $n$.
Our final set of authorized circuits are of this form.
We say that a circuit $C$ is a {\em $\PPA$-circuit} if for some $D$ and $F$, we have
$C = C_{D,F} \oplus C'$,
where $\cdeg(C') < n$. 

In computational problems considered in this paper, we assume that a 
$\PPA$-circuit $C = C_{D,F} \oplus C'$ is expicitly specified
by the circuits $D$, $F$ and $C'$.

\begin{lemma}
\label{lem:matching}
If $C$ is a~$\PPA$-circuit then $\mdeg(C(x)) < n$, and a perfect matching $\mu$ between the maximal parse subcircuits of $C$ can
be computed in polynomial time.
\end{lemma}

\begin{proof}
Let $C = C_{D,F} \oplus C'$ where $\mdeg(C') < n$. We can suppose without less of generality that $C'$ is the empty circuit,
that is $C = C_{D,F}$.
Since the six components of $C$ are pairwise disjoint (except for the input gates), every maximal parse subcircuit in $C$
consists of the mark of the root of $C$ from the set $\{1, \ldots, 6\}$, 
and a 
maximal parse subcircuit in the corresponding component. 
For the definition of $\mu$ we decompose $C$
into the disjoint sum of three circuits $C_1, C_2$ and $C_3$ where each of them is the disjoint sum of two  $\PPA$-components,
and will define the matching  
inside each of these circuits. The three circuits are as follows:
$C_1 = I_1 \diamond D_1 \circ F_1 \oplus I_2 \diamond F_2 \circ D_2$, 
$C_2 = I_3 \diamond D_3 \circ D_4 \oplus I_4 \diamond D_5$, and 
$C_3 = I_5 \diamond F_3 \circ F_4 \oplus I_6 \diamond F_5$. Clearly $C_2$ and $C_3$ are similar,
therefore it is sufficient
to define $\mu$ for $C_1$ and $C_2$.

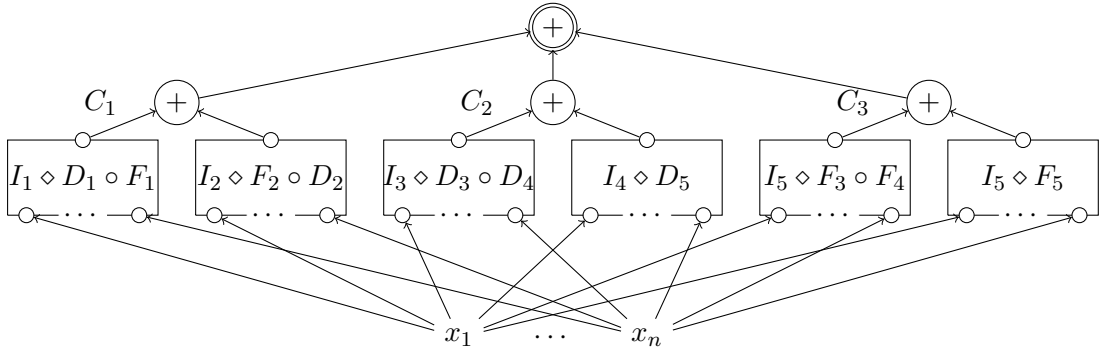
\begin{figure}[H]
	\centering
\begin{tikzpicture}[scale = 0.5,
	output/.style = {double, double distance = 1pt, circle, draw},
	gate/.style = {circle, draw},
	input/.style = {},
	small/.style = {minimum size = 2pt, inner sep = 2pt},
	ultra small/.style = {minimum size = 1pt, inner sep = 1pt}]
	\node [output, small] (s) at (12.5, 3) {$+$};
	\node [gate, small, label = {[label distance = 10pt]left:$C_1$}] (c1) at (2.5, 1) {$+$};
	\node [gate, small, label = {[label distance = 10pt]left:$C_2$}] (c2) at (12.5, 1) {$+$};
	\node [gate, small, label = {[label distance = 10pt]left:$C_3$}] (c3) at (22.5, 1) {$+$};

	\draw [->] (c1) -- (s);
	\draw [->] (c2) -- (s);
	\draw [->] (c3) -- (s);

	\begin{scope}[shift = {(0, 0.75)}]
		\node (x1) at (10, -6) {$x_1$};
		\node at (12.5, -6) {$\cdots$};
		\node (xn) at (15, -6) {$x_n$};
	\end{scope}

		\draw (-2, -2) -- (-2, 0) -- (2, 0) -- (2, -2) -- (-2, -2);
		\node at (0, -1) {$I_1 \diamond D_1 \circ F_1$};
		\node [gate, small, fill = white] (c) at (0, 0) {};
		\draw [->] (c) -- (c1);
		\node [gate, ultra small, draw = white, fill = white] at (0, -2) {$\cdots$};
		\node [gate, small, fill = white] (d1) at (-1.5, -2) {};
		\node [gate, small, fill = white] (dn) at (1.5, -2) {};
		\draw [->] (x1) -- (d1);
		\draw [->] (xn) -- (dn);

	\begin{scope}[shift = {(5, 0)}]
		\draw (-2, -2) -- (-2, 0) -- (2, 0) -- (2, -2) -- (-2, -2);
		\node at (0, -1) {$I_2 \diamond F_2 \circ D_2$};
		\node [gate, small, fill = white] (c) at (0, 0) {};
		\draw [->] (c) -- (c1);
		\node [gate, ultra small, draw = white, fill = white] at (0, -2) {$\cdots$};
		\node [gate, small, fill = white] (d1) at (-1.5, -2) {};
		\node [gate, small, fill = white] (dn) at (1.5, -2) {};
		\draw [->] (x1) -- (d1);
		\draw [->] (xn) -- (dn);
	\end{scope}

	\begin{scope}[shift = {(10, 0)}]
		\draw (-2, -2) -- (-2, 0) -- (2, 0) -- (2, -2) -- (-2, -2);
		\node at (0, -1) {$I_3 \diamond D_3 \circ D_4$};
		\node [gate, small, fill = white] (c) at (0, 0) {};
		\draw [->] (c) -- (c2);
		\node [gate, ultra small, draw = white, fill = white] at (0, -2) {$\cdots$};
		\node [gate, small, fill = white] (d1) at (-1.5, -2) {};
		\node [gate, small, fill = white] (dn) at (1.5, -2) {};
		\draw [->] (x1) -- (d1);
		\draw [->] (xn) -- (dn);
	\end{scope}

	\begin{scope}[shift = {(15, 0)}]
		\draw (-2, -2) -- (-2, 0) -- (2, 0) -- (2, -2) -- (-2, -2);
		\node at (0, -1) {$I_4 \diamond D_5$};
		\node [gate, small, fill = white] (c) at (0, 0) {};
		\draw [->] (c) -- (c2);
		\node [gate, ultra small, draw = white, fill = white] at (0, -2) {$\cdots$};
		\node [gate, small, fill = white] (d1) at (-1.5, -2) {};
		\node [gate, small, fill = white] (dn) at (1.5, -2) {};
		\draw [->] (x1) -- (d1);
		\draw [->] (xn) -- (dn);
	\end{scope}

	\begin{scope}[shift = {(20, 0)}]
		\draw (-2, -2) -- (-2, 0) -- (2, 0) -- (2, -2) -- (-2, -2);
		\node at (0, -1) {$I_5 \diamond F_3 \circ F_4$};
		\node [gate, small, fill = white] (c) at (0, 0) {};
		\draw [->] (c) -- (c3);
		\node [gate, ultra small, draw = white, fill = white] at (0, -2) {$\cdots$};
		\node [gate, small, fill = white] (d1) at (-1.5, -2) {};
		\node [gate, small, fill = white] (dn) at (1.5, -2) {};
		\draw [->] (x1) -- (d1);
		\draw [->] (xn) -- (dn);
	\end{scope}

	\begin{scope}[shift = {(25, 0)}]
		\draw (-2, -2) -- (-2, 0) -- (2, 0) -- (2, -2) -- (-2, -2);
		\node at (0, -1) {$I_5 \diamond F_5$};
		\node [gate, small, fill = white] (c) at (0, 0) {};
		\draw [->] (c) -- (c3);
		\node [gate, ultra small, draw = white, fill = white] at (0, -2) {$\cdots$};
		\node [gate, small, fill = white] (d1) at (-1.5, -2) {};
		\node [gate, small, fill = white] (dn) at (1.5, -2) {};
		\draw [->] (x1) -- (d1);
		\draw [->] (xn) -- (dn);
	\end{scope}

\end{tikzpicture}
	\caption{The decomposition $C = C_1 \oplus C_2 \oplus C_3$.}
	\label{fig:decomposition}
\end{figure}

\noindent
{\bf The matching $\mu$ inside $C_1$}.

To ease the notation, we rename the subcircuits of $C_1$ as
$I \diamond D \circ F$ and $I' \diamond F' \circ D'$, and we suppose that $I \diamond D \circ F$ is the left subcircuit
of $C_1$ and  $I' \diamond F' \circ D'$ is its right subcircuit. Let us denote the output (sum) gate of $C_1$ by $h$, 
the sum gates of $I$ by $h_1, \ldots, h_n$, the output gates of $D$ by $d_1, 
\ldots d_n$, and the output gates of $F$
by $f_1, \ldots, f_n$. 
For every gate $g$ in $I, D$ and $F$, we denote the corresponding 
gate in 
$I', D'$ and $F'$ by $g'$, and we also set $h' = h$. 
Let us recall the $h_i$ has three children, the left child is the input gate $x_i$, the center child is $d_i$,
the $i$th output gate of $D$, 
and its right child is the constant gate 1. A parse subcircuit can map $h_i$ into one of the three
marks $\ell, c$ and $r$, corresponding respectively to its left, center, and right child.

We define $\mu(S)$ for 
the maximal parse subcircuits of $I \diamond D \circ F$,
that is when $S(h) = \ell$. The definition for the case $S(h) = r$ is symmetric. 
Let us first define three sets of indices $S_{\out}, S_{\mi}, S_{\inp} \subseteq [n]$.
Let $S_{\out} = \{ i \in [n] : S(h_i) = c\}$, that is $S_{\out}$ contains those indices $i$ for
which the edge from the $d_i$ to $h_i$ belongs to $S$.
By definition $ i \in S_{\mi}$ if there exists an edge in $S$ from $f_i$ to a gate in $D$.
Finally, $ i \in S_{\inp}$ if there exists an edge in $S$ from $x_i$ to a gate in $F$.
We claim that $S_{\out} \subseteq S_{\inp}$.
This is indeed true, since if there exists $i \in S_{\out} \setminus S_{\inp}$ then the 
monomial $m_S(x)$ wouldn't contain the variable $x_i$, contradicting its maximality.
We are now ready to define $S' = \mu(S)$ by distinguishing two cases, depending on if $S_{\out}$ is a proper subset of $S_{\inp}$ or not.

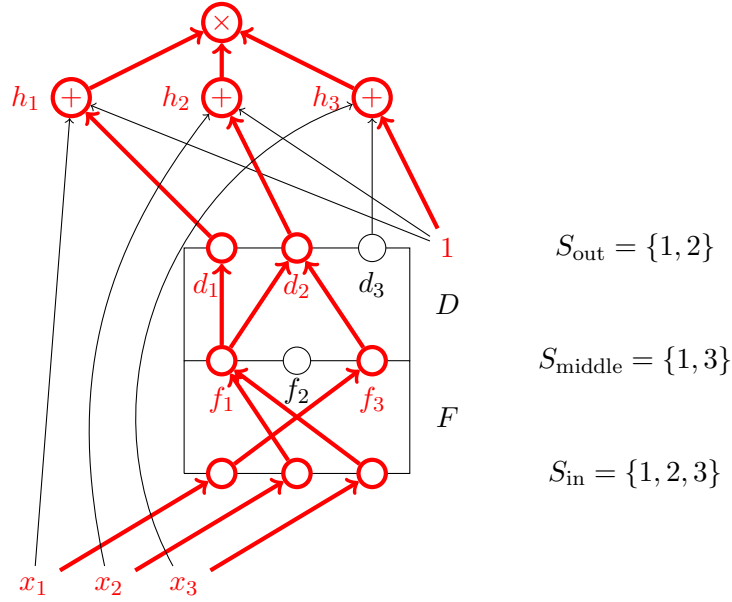
\begin{figure}[H]
	\centering
	\begin{tikzpicture}
		[gate/.style = {circle, draw, fill = white},
		parse/.style = {red, ultra thick},
		small/.style = {minimum size = 1pt, inner sep = 1pt}]
		\draw (0, 0) -- (3, 0) -- (3, 1.5) -- (0, 1.5) -- (0, 0);
		\draw (3, 1.5) -- (3, 3) -- (0, 3) -- (0, 1.5);

		\node [parse] (x1) at (-2, -1.5) {$x_1$};
		\node [parse] (x2) at (-1, -1.5) {$x_2$};
		\node [parse] (x3) at (0, -1.5) {$x_3$};

		\node [parse, gate] (v1) at (0.5, 0) {};
		\node [parse, gate] (v2) at (1.5, 0) {};
		\node [parse, gate] (v3) at (2.5, 0) {};

		\node (c1) at (3.5, 3) {\color{red}$1$};
		\node [parse, gate, small, label = left:{\color{red}$h_1$}] (i1) at (-1.5, 5) {$+$};
		\node [parse, gate, small, label = left:{\color{red}$h_2$}] (i2) at (0.5, 5) {$+$};
		\node [parse, gate, small, label = left:{\color{red}$h_3$}] (i3) at (2.5, 5) {$+$};

		\node [parse, gate, small] (t) at (0.5, 6) {$\times$};

		\draw [->, parse] (i1) -- (t);
		\draw [->, parse] (i2) -- (t);
		\draw [->, parse] (i3) -- (t);

		\draw [->, parse] (x1) -- (v1);
		\draw [->, parse] (x2) -- (v2);
		\draw [->, parse] (x3) -- (v3);

		\node at (3.5, 2.25) {$D$};
		\node at (3.5, 0.75) {$F$};

		\node [parse, gate, label = below:{\color{red}$f_1$}] (f1) at (0.5, 1.5) {};
		\node [gate, label = {[shift = {(0, -0.9)}]$f_2$}] (f2) at (1.5, 1.5) {};
		\node [parse, gate, label = below:{\color{red}$f_3$}] (f3) at (2.5, 1.5) {};

		\node [parse, gate, label = {[shift = {(-0.2, -1)}]{\color{red}$d_1$}}] (d1) at (0.5, 3) {};
		\node [parse, gate,  label = below:{\color{red}$d_2$}] (d2) at (1.5, 3) {};
		\node [gate, label = below:$d_3$] (d3) at (2.5, 3) {};

		\draw [->, parse] (v1) -- (f3);
		\draw [->, parse] (v2) -- (f1);
		\draw [->, parse] (v3) -- (f1);

		\draw [->, parse] (f1) -- (d1);
		\draw [->, parse] (f1) -- (d2);
		\draw [->, parse] (f3) -- (d2);

		\draw [->] (x1) -- (i1);
		\draw [->, parse] (d1) -- (i1);
		\draw [->] (c1) -- (i1);

		\draw [->] (x2) [bend left = 25] to (i2);
		\draw [->, parse] (d2) -- (i2);
		\draw [->] (c1) -- (i2);

		\draw [->] (x3) [bend left = 50] to (i3);
		\draw [->] (d3) -- (i3);
		\draw [->, parse] (c1) -- (i3);

		\node at (6, 0) {$S_{\textrm{in}} = \{1, 2, 3\}$};
		\node at (6, 1.5) {$S_{\textrm{middle}} = \{1, 3\}$};
		\node at (6, 3) {$S_{\textrm{out}} = \{1, 2\}$};
	\end{tikzpicture}
	\caption{The left subcircuit $I \diamond D \circ F$ of $C_1$ and the index sets $S_{\textrm{in}}, S_{\textrm{middle}}$ and $S_{\textrm{out}}$. 
	}
	\label{fig:indexsets}
\end{figure}

\noindent
{\em Case $1$}: $S_{\out} \subset S_{\inp}$.

Let $i$ be the smallest index in $S_{\inp} \setminus S_{\out}$. By definition, we let $S'$ be the same as $S$, except
on $h_i$, where $S'$ takes the mark $r$ when $S(h_i) = \ell$, and it takes the mark $\ell$ when $S(h_i) = r$.
This means that the only difference between $S$ and $S'$ is that at the
$i$th sum gate of $I$, one subcircuit contains the edge from $x_i$ to $h_i$, 
whereas the other contains the edge from 1 to $h_i$. 
$S'$ is therefore a parse subcircuit.
To show that $S'$ is also maximal,
the interesting case is when $S(h_i) = \ell$ and $S'(h_i) = r$, that is $m_{S'}(x)$ doesn't 
directly pick up $x_i$ at $h_i$. But since $i \in S_{\inp}$,
the variable $x_i$ is still 
in $S'$, which is therefore 
maximal. Finally clearly  $\mu(S') = S$.

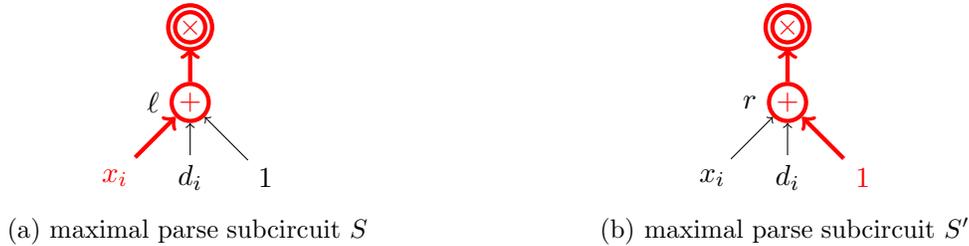
\begin{figure}[H]
	\centering
	\begin{subfigure}[b]{0.45\textwidth}
		\centering
\begin{tikzpicture}
	[output/.style = {double, double distance = 1pt, circle, draw},
	gate/.style = {circle, draw},
	input/.style = {},
	parse/.style = {color = red, ultra thick},
	small/.style = {minimum size = 1pt, inner sep = 1pt}]

	\node [output, parse, small] (s) at (0, 0) {$\times$};
	\node [gate, parse, small, label = left:$\ell$] (h) at (0, -1) {$+$};
	\node [input, parse] (x) at (-1, -2) {$x_i$};
	\node [input] (d) at (0, -2) {$d_i$};
	\node [input] (c) at (1, -2) {$1$};

	\draw [->, parse] (h) -- (s);
	\draw [->, parse] (x) -- (h);
	\draw [->] (d) -- (h);
	\draw [->] (c) -- (h);
\end{tikzpicture}
		\subcaption{maximal parse subcircuit $S$}
	\end{subfigure}
\quad
	\begin{subfigure}[b]{0.45\textwidth}
		\centering
\begin{tikzpicture}
	[output/.style = {double, double distance = 1pt, circle, draw},
	gate/.style = {circle, draw},
	input/.style = {},
	parse/.style = {color = red, ultra thick},
	small/.style = {minimum size = 1pt, inner sep = 1pt}]

	\node [output, parse, small] (s) at (0, 0) {$\times$};
	\node [gate, parse, small, label = left:$r$] (h) at (0, -1) {$+$};
	\node [input] (x) at (-1, -2) {$x_i$};
	\node [input] (d) at (0, -2) {$d_i$};
	\node [input, parse] (c) at (1, -2) {$1$};

	\draw [->, parse] (h) -- (s);
	\draw [->] (x) -- (h);
	\draw [->] (d) -- (h);
	\draw [->, parse] (c) -- (h);
\end{tikzpicture}
		\subcaption{maximal parse subcircuit $S'$}
	\end{subfigure}
	\caption{Case 1  of the matching $\mu$ for $C_1$ where $i$ is the smallest index in 
	$ S_{\textrm{in}} \setminus S_{\textrm{out}}$.}
\end{figure}

\noindent
{\em Case $2$}: $S_{\out} = S_{\inp}$.

In that case first observe that for every index 
$i \not \in S_{\out}$, we have $S(h_i) = \ell$, that is $S$ contains the edge $(x_i, h_i)$, since otherwise $m_S(x)$ wouldn't
contain $x_i$.
By definition, let $\Dom(S') = \{g' \in G^+ : g \in \Dom(S)\}$.
For the output gate $h' =h$ of $C_1$ we set $S'(h') =r$, that is $S'$ will
be a parse subcircuit of $I' \diamond D' \circ F'$. 
For the sum gates $h'_1, \ldots, h'_n$ of $I$, we set $S'(h'_i) = c$ if $i \in S_{\mi}$, and we set
$S'(h'_i) = \ell$ otherwise.
Finally, for every  sum gate $g \in \Dom(S)$ 
in $D$ or in $F$, we set $S'(g') = S(g)$.

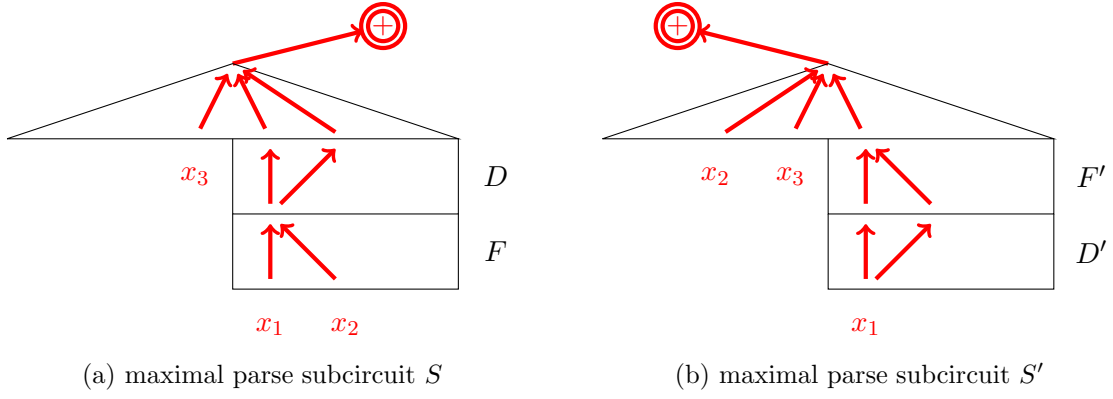
\begin{figure}[H]
	\centering
	\begin{subfigure}[b]{0.45\textwidth}
		\centering
\begin{tikzpicture}
	[output/.style = {double, double distance = 1pt, circle, draw},
	gate/.style = {circle, draw},
	input/.style = {},
	parse/.style = {color = red, ultra thick},
	small/.style = {minimum size = 1pt, inner sep = 1pt}]
	\node [output, parse, small] (s) at (2, 0.5) {$+$};

	\draw [->, parse] (0, 0) -- (s);
		\node at (3.5, -1.5) {$D$};
		\node at (3.5, -2.5) {$F$};
		\node at (-2.5, -1.5) {};
		\node at (-1.5, -1.5) {};
		\node [parse] at (-0.5, -1.5) {$x_3$};
		\node [parse] at (0.5, -3.5) {$x_1$};
		\node [parse] at (1.5, -3.5) {$x_2$};
		\node at (2.5, -3.5) {};
		\draw (-3, -1) -- (0, 0) -- (3, -1) -- (-3, -1);
		\draw (0, -1) -- (0, -2) -- (3, -2) -- (3, -1);
		\draw (0, -2) -- (0, -3) -- (3, -3) -- (3, -2);

	\node (o) at (0, 0) {};
	\node (l1a) at (-0.5, -1) {};
	\node (l1b) at (0.5, -1) {};
	\node (l1c) at (1.5, -1) {};
	\node (l2) at (0.5, -2) {};
	\node (l3a) at (0.5, -3) {};
	\node (l3b) at (1.5, -3) {};

	\draw [->, parse] (l1a) -- (o);
	\draw [->, parse] (l1b) -- (o);
	\draw [->, parse] (l1c) -- (o);
	\draw [->, parse] (l2) -- (l1b);
	\draw [->, parse] (l2) -- (l1c);
	\draw [->, parse] (l3a) -- (l2);
	\draw [->, parse] (l3b) -- (l2);

\end{tikzpicture}
		\subcaption{maximal parse subcircuit $S$}
	\end{subfigure}
	\quad
	\begin{subfigure}[b]{0.45\textwidth}
		\centering
\begin{tikzpicture}
	[output/.style = {double, double distance = 1pt, circle, draw},
	gate/.style = {circle, draw},
	input/.style = {},
	parse/.style = {color = red, ultra thick},
	small/.style = {minimum size = 1pt, inner sep = 1pt}]
	\node [output, parse, small] (s) at (-2, 0.5) {$+$};

	\draw [->, parse] (0, 0) -- (s);
		\node at (3.5, -1.5) {$F'$};
		\node at (3.5, -2.5) {$D'$};
		\node at (-2.5, -1.5) {};
		\node [parse] at (-1.5, -1.5) {$x_2$};
		\node [parse] at (-0.5, -1.5) {$x_3$};
		\node [parse] at (0.5, -3.5) {$x_1$};
		\node at (1.5, -3.5) {};
		\node at (2.5, -3.5) {};
		\draw (-3, -1) -- (0, 0) -- (3, -1) -- (-3, -1);
		\draw (0, -1) -- (0, -2) -- (3, -2) -- (3, -1);
		\draw (0, -2) -- (0, -3) -- (3, -3) -- (3, -2);

	\node (o) at (0, 0) {};
	\node (l1a) at (-1.5, -1) {};
	\node (l1b) at (-0.5, -1) {};
	\node (l1c) at (0.5, -1) {};
	\node (l2a) at (0.5, -2) {};
	\node (l2b) at (1.5, -2) {};
	\node (l3) at (0.5, -3) {};

	\draw [->, parse] (l1a) to 
	(o);
	\draw [->, parse] (l1b) -- (o);
	\draw [->, parse] (l1c) -- (o);
	\draw [->, parse] (l2a) -- (l1c);
	\draw [->, parse] (l2b) -- (l1c);
	\draw [->, parse] (l3) -- (l2a);
	\draw [->, parse] (l3) -- (l2b);

\end{tikzpicture}
		\subcaption{maximal parse subcircuit $S'$}
	\end{subfigure}
	\caption{Case 2 of the matching $\mu$ for $C_1$:
	$S_{\textrm{out}} = S_{\textrm{in}}$.}
\end{figure}

Let us recall that $V_S$ is the set of vertices of the accessibility graph $G_S$ of $S$.
The proof that $S'$ is a maximal parse subcircuit immediately follows from the following proposition.

\begin{proposition}
\label{prop:accessibility}
For every computational gate 
$g$ in $I \diamond D \circ F$, we 
have 
$$g \in V_S \mbox{{\it ~ if and only if ~}} g' \in V_{S'}.$$ 
\end{proposition}

\begin{proof}
We show the implication from left to right.
This is certainly true for the 
computational gates of $I$ since they are all accessible in $G_S$, as well as the 
computational gates of $I'$ in $G_{S'}$.

If $g \in  V_S$ is a computational gate of $D$ then there is a path $p$ in $G_S$ from $g$ to $h$ which can be decomposed into
$p=p_1p_2$, where $p_1$ goes from $g$ to $d_i$ for some $i \in S_{\out}$, and $p_2$ is the path from $d_i$
to $h$.
In $G_{S'}$ we have therefore a path $p'_1$ from $g'$ to $d'_i$. Since $S_{\out} = S_{\inp}$, in $G_{S}$ we have a path
$p_3$ from $x_i$ to $f_j$ for some $j \in S_{\mi}$. Therefore in $G_{S'}$ there exists a path $p'_2$ from 
$d'_i$ to $f'_j$. Finally, in $G_{S'}$ there is also a path $p_3'$
from $f'_j$ to $h'$ because $j \in S_{\mi}$.
Then $p' = p'_1p'_2p'_3$ is a path from $g'$ to $h'$. 

If $g \in  V_S$ is a computational gate of $F$ then there is a path $p$ in $G_S$ from $g$ to $h$ which can be decomposed into
$p=p_1p_2p_3$, where $p_1$ goes from $g$ to $d_i$ for some $i \in S_{\mi}$, 
$p_2$ goes from $d_i$ to $f_j$ for some $j \in S_{\out}$,
and $p_3$ is the path from $f_j$
to $h$. Then in $G_{S'}$ there exists a path $p'_1$ from $g'$ to $d'_i$,
and a path $p'_2$ which goes from $d'_i$ to $h'$ since $i \in S_{\mi}$. Then the path
$p' = p'_1p'_2$ goes from $g'$ to $h'$.


The implication from right to left follows from the symmetry 
between $S$ and $S'$. For this, it is useful to observe that $S'_{\out} = S'_{\inp} = S_{\mi},$ and
$S'_{\mi} = S_{\out} =S_{\inp}.$ 
\end{proof}

We have $\Dom(S) = V_{S} \cap G^+$ since $S$ is a parse subcircuit. 
Proposition~\ref{prop:accessibility} and the definition $\Dom(S') = \{g' \in G^+ : g \in \Dom(S)\}$  imply that 
$\Dom(S') = V_{S'} \cap G^+$, and therefore
$S'$ is a parse subcircuit. To prove the 
maximality of $S'$ 
let us show that every input gate is in $V_{S'}$. If $i \in S_{\mi}$ 
then the path $p$ defined above for the computational
gates in $D$ yields a path $p'$ from $x_i$ to $h'$.
If $i \not \in S_{\mi}$ then the direct path $p'$ from $x_i$ to $h'$ via $h_i'$
exists in $G_{S'}$.
Finally $\mu$ is
clearly  involutive in that case too. 

\bigskip
\noindent
{\bf The matching $\mu$ inside $C_2$}.

We now turn to the description of $\mu$ for $C_2$, where we rename its two subcircuits as
$I \diamond D \circ D'$ and $I^* \diamond D^*$.
The matching for $C_2$ has strong analogies with the matching for $C_1$, 
to better see this
we also 
use the names $I' ,F$ and $F'$ respectively for the circuits $I, D'$ and $D$.
This means that $I \diamond D \circ F$
and 
$I' \diamond F' \circ D'$ 
are just different names for the circuit $I \diamond D \circ D'$.
We suppose that $I \diamond D \circ D'$ is the left subcircuit
of $C_2$ and  $I^*  \diamond D^*$ is its right subcircuit. 
Similarly to the circuit $C_1$,
we denote the output gate of $C_2$ by $h$, 
the sum gates of $I$ by $h_1, \ldots, h_n$, 
the ouput gates of $D$ by $d_1, \ldots d_n$,
and the output gates of $D'$ by $d'_1, \ldots, d'_n$.
For every gate $g$ in $I, D$ and $D'$, we denote the corresponding gate respectively in 
$I', D'$ 
and $D$ by $g'$.
For every gate $g$ in $I$ and $D$, we denote the corresponding gate in $I^*$ and
$D^*$ by $g^*$. 
We also set $h^* = h' = h$. 
Again, $h_i$ has three children, the left child is the input gate $x_i$, the center child is $d_i$, 
the right child is the constant gate 1, and the respective marks are $\ell, c$ and $r$.

We first describe $S' = \mu(S)$ when $S$ is a maximal parse subcircuit of $I \diamond D \circ D'$.
We define $S_{\out}, S_{\mi}, S_{\inp}$ the same way as for the circuit $I \diamond D \circ F$,
keeping in mind that $F =D'$.
As before, we have $S_{\out} \subseteq S_{\inp}$. For the definition of $\mu$ we 
now distinguish three cases.

\bigskip
\noindent
{\em Case $1$}: $S_{\out} \subset S_{\inp}$.

The definition of $S'$ is identical to the first case of the definition of the matching for $C_1$.

%

\bigskip
\noindent
{\em Case $2$}: $S_{\out} = S_{\inp}$ and there exists a sum gate $g$ in $D$ such that $S(g) \neq S(g')$.

The definition of $S'$ is identical to the second case of the definition of the matching for $C_1$, with 
one exception. The difference is that $S'$ remains in the left subcircuit of 
$C_2$, that is 
for the output gate $h' =h$ we set
$S'(h') = \ell$.

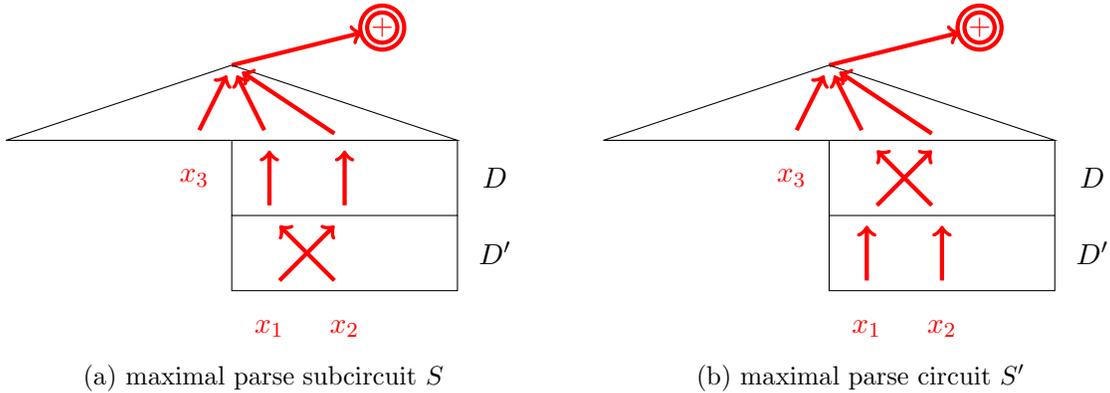
\begin{figure}[H]
	\centering
	\begin{subfigure}[b]{0.45\textwidth}
		\centering
\begin{tikzpicture}
	[output/.style = {double, double distance = 1pt, circle, draw},
	gate/.style = {circle, draw},
	input/.style = {},
	parse/.style = {color = red, ultra thick},
	small/.style = {minimum size = 1pt, inner sep = 1pt}]
	\node [output, parse, small] (s) at (2, 0.5) {$+$};

	\draw [->, parse] (0, 0) -- (s);
		\node at (3.5, -1.5) {$D$};
		\node at (3.5, -2.5) {$D'$};
		\node at (-2.5, -1.5) {};
		\node at (-1.5, -1.5) {};
		\node [parse] at (-0.5, -1.5) {$x_3$};
		\node [parse] at (0.5, -3.5) {$x_1$};
		\node [parse] at (1.5, -3.5) {$x_2$};
		\node at (2.5, -3.5) {};
		\draw (-3, -1) -- (0, 0) -- (3, -1) -- (-3, -1);
		\draw (0, -1) -- (0, -2) -- (3, -2) -- (3, -1);
		\draw (0, -2) -- (0, -3) -- (3, -3) -- (3, -2);

	\node (o) at (0, 0) {};
	\node (l1a) at (-0.5, -1) {};
	\node (l1b) at (0.5, -1) {};
	\node (l1c) at (1.5, -1) {};
	\node (l2a) at (0.5, -2) {};
	\node (l2b) at (1.5, -2) {};
	\node (l3a) at (0.5, -3) {};
	\node (l3b) at (1.5, -3) {};

	\draw [->, parse] (l1a) -- (o);
	\draw [->, parse] (l1b) -- (o);
	\draw [->, parse] (l1c) -- (o);
	\draw [->, parse] (l2a) -- (l1b);
	\draw [->, parse] (l2b) -- (l1c);
	\draw [->, parse] (l3a) -- (l2b);
	\draw [->, parse] (l3b) -- (l2a);

\end{tikzpicture}
		\subcaption{maximal parse subcircuit $S$}
	\end{subfigure}
	\quad
	\begin{subfigure}[b]{0.45\textwidth}
		\centering
\begin{tikzpicture}
	[output/.style = {double, double distance = 1pt, circle, draw},
	gate/.style = {circle, draw},
	input/.style = {},
	parse/.style = {color = red, ultra thick},
	small/.style = {minimum size = 1pt, inner sep = 1pt}]
	\node [output, parse, small] (s) at (2, 0.5) {$+$};

	\draw [->, parse] (0, 0) -- (s);
		\node at (3.5, -1.5) {$D$};
		\node at (3.5, -2.5) {$D'$};
		\node at (-2.5, -1.5) {};
		\node at (-1.5, -1.5) {};
		\node [parse] at (-0.5, -1.5) {$x_3$};
		\node [parse] at (0.5, -3.5) {$x_1$};
		\node [parse] at (1.5, -3.5) {$x_2$};
		\node at (2.5, -3.5) {};
		\draw (-3, -1) -- (0, 0) -- (3, -1) -- (-3, -1);
		\draw (0, -1) -- (0, -2) -- (3, -2) -- (3, -1);
		\draw (0, -2) -- (0, -3) -- (3, -3) -- (3, -2);

	\node (o) at (0, 0) {};
	\node (l1a) at (-0.5, -1) {};
	\node (l1b) at (0.5, -1) {};
	\node (l1c) at (1.5, -1) {};
	\node (l2a) at (0.5, -2) {};
	\node (l2b) at (1.5, -2) {};
	\node (l3a) at (0.5, -3) {};
	\node (l3b) at (1.5, -3) {};

	\draw [->, parse] (l1a) -- (o);
	\draw [->, parse] (l1b) -- (o);
	\draw [->, parse] (l1c) -- (o);
	\draw [->, parse] (l2a) -- (l1c);
	\draw [->, parse] (l2b) -- (l1b);
	\draw [->, parse] (l3a) -- (l2a);
	\draw [->, parse] (l3b) -- (l2b);

\end{tikzpicture}
		\subcaption{maximal parse circuit $S'$}
	\end{subfigure}
	\caption{Case 2 of the matching $\mu$ for $C_2$:
	$S_{\textrm{out}} = S_{\textrm{in}}$ and $\exists g, S(g) \ne S(g').$
	}
\end{figure}



\noindent
{\em Case $3$}: $S_{\out} = S_{\inp}$ and for all sum gate $g$ in $D$, we have $S(g) = S(g')$.

By definition we set $\Dom(S') = \{g^* \in G^+ : g \in \Dom(S)\}$.
For the output gate $h^* =h$ of $C_2$ we set $S'(h^*) =r$, that is $S'$ will
be a parse subcircuit of $I^* \diamond D^*$.  For every other 
sum gate $g \in \Dom(S)$, we set $S'(g^*) = S(g)$. 

The description $S' = \mu(S)$ when $S$ is a maximal parse subcircuit of $I^* \diamond D^*$ is as follows.
By definition we set $\Dom(S')  = 
\{g,g' \in G^+ : g^* \in \Dom(S)\}$. 
We set $S'(h) = \ell$, that is $S'$ is a parse subcircuit of $I \diamond D \circ D'.$ For the sum gates of $I$, we set
$S'(h_i) = S(h^*_i)$. For every sum gate $g^* \in \Dom(S)$ which is in $D^*$, we set $S'(g) = S'(g')=S(g^*)$.

\begin{figure}[H]
	\centering
	\begin{subfigure}[b]{0.45\textwidth}
		\centering
\begin{tikzpicture}
	[output/.style = {double, double distance = 1pt, circle, draw},
	gate/.style = {circle, draw},
	input/.style = {},
	parse/.style = {color = red, ultra thick},
	small/.style = {minimum size = 1pt, inner sep = 1pt}]
	\node [output, parse, small] (s) at (2, 0.5) {$+$};

	\draw [->, parse] (0, 0) -- (s);
		\node at (3.5, -1.5) {$D$};
		\node at (3.5, -2.5) {$D'$};
		\node at (-2.5, -1.5) {};
		\node at (-1.5, -1.5) {};
		\node [parse] at (-0.5, -1.5) {$x_3$};
		\node [parse] at (0.5, -3.5) {$x_1$};
		\node [parse] at (1.5, -3.5) {$x_2$};
		\node at (2.5, -3.5) {};
		\draw (-3, -1) -- (0, 0) -- (3, -1) -- (-3, -1);
		\draw (0, -1) -- (0, -2) -- (3, -2) -- (3, -1);
		\draw (0, -2) -- (0, -3) -- (3, -3) -- (3, -2);

	\node (o) at (0, 0) {};
	\node (l1a) at (-0.5, -1) {};
	\node (l1b) at (0.5, -1) {};
	\node (l1c) at (1.5, -1) {};
	\node (l2a) at (0.5, -2) {};
	\node (l2b) at (1.5, -2) {};
	\node (l3a) at (0.5, -3) {};
	\node (l3b) at (1.5, -3) {};

	\draw [->, parse] (l1a) -- (o);
	\draw [->, parse] (l1b) -- (o);
	\draw [->, parse] (l1c) -- (o);
	\draw [->, parse] (l2a) -- (l1b);
	\draw [->, parse] (l2b) -- (l1c);
	\draw [->, parse] (l3a) -- (l2a);
	\draw [->, parse] (l3b) -- (l2b);

\end{tikzpicture}
		\subcaption{maximal parse subcircuit $S$}
	\end{subfigure}
	\quad
	\begin{subfigure}[b]{0.45\textwidth}
		\centering
\begin{tikzpicture}
	[output/.style = {double, double distance = 1pt, circle, draw},
	gate/.style = {circle, draw},
	input/.style = {},
	parse/.style = {color = red, ultra thick},
	small/.style = {minimum size = 1pt, inner sep = 1pt}]
	\node [output, parse, small] (s) at (-2, 0.5) {$+$};

	\draw [->, parse] (0, 0) -- (s);
		\node at (3.5, -1.5) {$D^*$};
		\node at (-2.5, -1.5) {};
		\node at (-1.5, -1.5) {};
		\node [parse] at (-0.5, -1.5) {$x_3$};
		\node [parse] at (0.5, -2.5) {$x_1$};
		\node [parse] at (1.5, -2.5) {$x_2$};
		\node at (2.5, -2.5) {};
		\draw (-3, -1) -- (0, 0) -- (3, -1) -- (-3, -1);
		\draw (0, -1) -- (0, -2) -- (3, -2) -- (3, -1);

	\node (o) at (0, 0) {};
	\node (l1a) at (-0.5, -1) {};
	\node (l1b) at (0.5, -1) {};
	\node (l1c) at (1.5, -1) {};
	\node (l2a) at (0.5, -2) {};
	\node (l2b) at (1.5, -2) {};

	\draw [->, parse] (l1a) -- (o);
	\draw [->, parse] (l1b) -- (o);
	\draw [->, parse] (l1c) -- (o);
	\draw [->, parse] (l2a) -- (l1b);
	\draw [->, parse] (l2b) -- (l1c);

\end{tikzpicture}
		\subcaption{maximal parse subcircuit $S'$}
	\end{subfigure}
	\caption{Case 3 of the matching $\mu$ for $C_2$:
	$S_{\textrm{out}} = S_{\textrm{in}}$ and 
	$\forall g, S(g) = S(g')$.}
\end{figure}
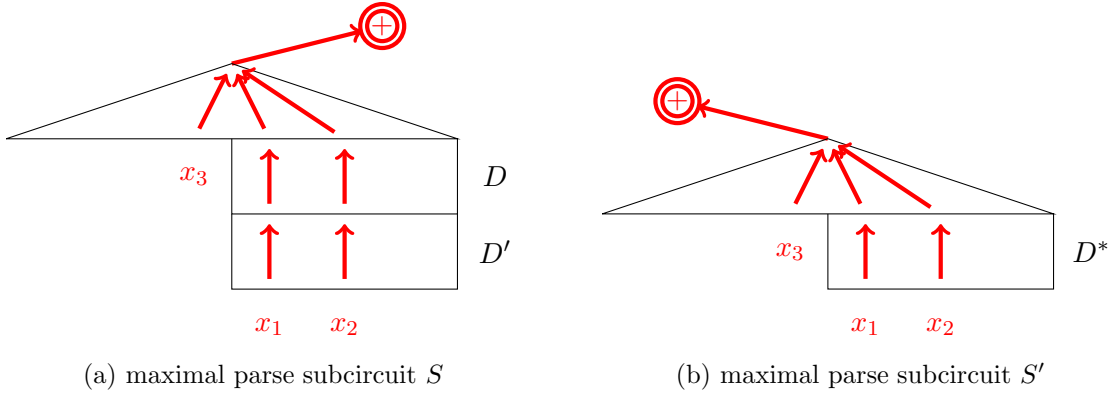

The proof that $S'$ is a maximal parse subcircuit is basically the same as for the case of circuit $C_1$.
It follows immediately from the definition that $\mu$ is an involution.
The only additional point to see is that in the second case $S' \neq S$ 
because $S(g) \neq S(g')$, for some gate $g$ in $D$.
\end{proof}

\section{The computational problems}
\label{sec:problems}
We are now ready to define \textsc{PPA-Circuit CNSS} and \textsc{PPA-Circuit Chevalley},
the two computational problems corresponding to the CNSS and to the
Chevalley-Warning
theorem 
over $\F_2$. The input will be in both cases an $n$-variable, single-output 
PPA-circuit $C$, and an element $a \in \F_2^n$.
In the case of \textsc{PPA-Circuit Chevalley}, it is a zero of $C$, and Lemma~\ref{lem:matching} ensures
that $C$ satisfies the hypotheses of the Chevalley-Warning Theorem.
For \textsc{PPA-Circuit CNSS}, we consider the circuit $C \oplus L_a$, and  Lemma~\ref{lem:matching} again
ensures that this circuit satisfies the hypothesis of the CNSS. The computational task is to compute $b \in \F_2^n$
whose existence is stipulated by these theorems.


The definition of the two problems is the following.

\begin{quote}
\textsc{PPA-Circuit Chevalley}

\textit{Input:} $(C, a)$, where  $C$ is an $n$-variable PPA-circuit over $\F_2$, and
$a$ is a zero of $C$.

\textit{Output:} Another zero $b \neq a$ of $C$.

\end{quote}

\begin{quote}
\textsc{PPA-Circuit CNSS}

\textit{Input:} $(C', a)$, where $C'$ is an $n$-variable PPA-circuit over $\F_2$, and $a \in \F_2^n$.


\textit{Output:} An element $b \in \F_2^n$ satisfying $C = C' \oplus L_a$.
\end{quote}

We stress that that a $\PPA$-circuit $C = C_{D,F} \oplus C_0$
is input as a triple consisting of two
$n$-input $n$-output circuits $D$ and $F$
and an $n$-input single-output $C_0$.

Let us restate here our main theorem. 

\noindent
{\bf Theorem 1.}
{\em The problems \textsc{PPA-Circuit CNSS} and \textsc{PPA-Circuit Chevalley} are $\PPA$-complete.} 

\begin{proof}
In Proposition~\ref{prop:equivalent} below we show that \textsc{PPA-Circuit CNSS} and \textsc{PPA-Circuit Chevalley} are polynomially
interreducible. In 
{ Theorem~\ref{thm:easiness}} in {Section~\ref{sec:easiness}} we prove that \textsc{PPA-Circuit CNSS} is in PPA, and in
{Theorem~\ref{thm:hardness}} in
{Section~\ref{sec:hardness}} we prove that \textsc{PPA-Circuit Chevalley} is PPA-hard.
\end{proof}

We now turn to the proof of the various parts of Theorem~\ref{thm:main}.

\begin{proposition}
\label{prop:equivalent}
\textsc{PPA-Circuit CNSS} and \textsc{PPA-Circuit Chevalley} are polynomially equivalent. 

\end{proposition}

\begin{proof}
First we reduce \textsc{PPA-Circuit CNSS} to \textsc{PPA-Circuit Chevalley}. Let 
$(C', a)$ 
be an instance of \textsc{PPA-Circuit CNSS}, and set $C = C' \oplus L_a$.
We can suppose that $C'(a) = 1$, since otherwise we are done. We define the 
circuit $C'' = C \oplus 1$. 
Then clearly $C''$ is a PPA-circuit, and $C''(a) = 0$.
The result of the reduction is then the input $(C'',a)$ to \textsc{PPA-Circuit Chevalley}. If the solution to that input
is another zero $b \neq a$ of $C''(x)$, then clearly $C(b) = 1$. 

The reduction from \textsc{PPA-Circuit Chevalley} 
to \textsc{PPA-Circuit CNSS} is very similar.
Let $(C,a)$ be an instance of \textsc{PPA-Circuit Chevalley}. We set 
$C' = C \oplus 1$, and $C'' = C' \oplus L_a$. Clearly $C'$ is a PPA-circuit. 
The result of the reduction is $(C',a)$. If the solution 
to that input is a satisfying assignment $C''(b) = 1$ then $b$ is a zero of $C$.
Also, $b \neq a$ since $C''(a) = 0$, therefore $b$ is another zero of $C$.
\end{proof}




\section{$\PPA$-easiness}
\label{sec:easiness}

\begin{theorem}
\label{thm:easiness}
\textsc{PPA-Circuit CNSS} is in $\PPA$.
\end{theorem}
\begin{proof}
We will give a reduction from \textsc{PPA-Circuit CNSS} to \textsc{Leaf}.
Given an input $N = (C',a)$ to \textsc{PPA-Circuit CNSS}, we set $C = C' \oplus L_a$.
We construct a graph $G_N = (V_N, E_N)$ by a polynomial time edge recognition algorithm and a 
polynomial time pairing function $\phi$ as explained in Section~\ref{sec:prelim_ppa}.
The vertices of $G_N$ are $V_N = \F_2^n \cup \cS(C)$.

There are two types of edges in $E_N$, 
the first type is between an assignment and a parse subcircuit, and the second type is between two maximal parse subcircuits.
By definition, the edge $\{a, S\}$ exists between $a \in \F_2^n$ 
and $S \in \cS(C)$
if $m_S(a) = 1$.  Such an edge can be easily recognized since the monomial $m_S(x)$ can be evaluated in linear time in the
size of $C$. 

Since $C$ is the disjoint sum of $C'$ and  $L_a$, the maximal parse subcircuits of $C$ are the maximal parse subcircuits of $C'$
extended with the appropriate mark at the output gate, and the unique maximal parse subcircuit of $L_a$, again extended with
the appropriate mark at the output gate. Let us denote the latter parse subcircuit by $T$.
Let $\mu$ be a polynomial time computable perfect matching between the maximal parse subcircuits of $C'$, 
which exists by Lemma~\ref{lem:matching}. 
By definition, the edge $\{S, S'\}$ exists between $S, S' \in \cS(C')$ 
if both are extensions of maximal parse subcircuits of $C'$, and their restrictions to $C'$ are matched by $\mu$. 

Observe that 
by Proposition~\ref{prop:char2}, a vertex $a \in \F_2^n$ 
has odd degree
if and only if $C(a) =1$. If $S$ is a maximal parse subcircuit then among the vertices
in $\F_2^n$
it is only connected to $1^n$. If $S \neq T$, then it has one more neighbor, its matching pair given by $\mu$, and therefore
its degree is two.
On the other hand, the degree of $T$ is one and therefore it is odd.
We can therefore take $T$ as the standard leaf.

We first give the pairing for the vertices in $\cS(C)$. We fix $S \in \cS(C)$, and let $a \in
\F_2^n$ 
such that $m_S(a) = 1$. If $S$ is not a maximal parse subcircuit then let $i \in [n]$ be the smallest integer such that
$x_i$ is not in $m_S(x)$, and let $a'$ be obtained from $a$ by flipping the $i$th bit. Then 
by definition $\phi(S, \cdot)$ pairs $a$ with $a'$. 
If $S \neq T$ is a maximal parse subcircuit then it has two neighbors: its matching pair $S'$ by $\mu$ and $1^n$,
and $\phi(S, \cdot)$ pairs these two neighbors. 
For every $S$, the mapping $\phi(S, \cdot)$
is clearly involutive.

We now turn to the more complicated pairing for the vertices in $\F_2^n$. 
Observe that this depends only on the edges of the first type, that is edges between an assignment 
$a \in \F_2^n$ and a parse subcircuit $S \in \cS(C)$. These edges can be defined actually for an arbitrary circuit $C$.
Let us denote by $G(C)$ the graph with vertex set $\F_2^n \cup \cS(C)$ and with edges of the first type from $G_N$.
First we prove the following lemma about $G(C)$ on induction of the size of $C$.
\begin{lemma}
\label{whatever}
For every $n$-variable, single-output circuit $C$, and for every vertex $a \in \F_2^n$ in $G(C)$,
\begin{enumerate}[a)]
\item 
if $\deg(a)$ is even then for all $S \in \cS(C)$ such that $m_S(a) = 1$, there exists $g \in \Dom(S)$ with $P_g(a) = 0$,
\item
if $\deg(a)$ is odd then there exists a unique $S \in \cS(C)$ such that $m_S(a) = 1$, and  $P_g(a) =  1$ for all $g \in \Dom(S)$.
\end{enumerate}
\end{lemma}

\begin{proof}
 If $C$ consists of a single node,
the statement is obviously true. Otherwise we first handle $a)$.
When $\deg(a)$ is even then $C(a) = 0$. If the root is a sum gate then 
we are done since it
is in the domain of every parse subcircuit. If the root is a product gate then at least one of its children (say the left without loss
of generality) also evaluates to $0$, that is ${C_{\ell}}(a) = 0$. Let $S \in \cS(C)$ be such that $m_S(a) = 1$, then we also have
$m_{S_{\ell}}(a) = 1$. By the inductive hypothesis there exists $g \in \Dom(S_{\ell})$ with $P_g(a) = 0$,
and since $g$ is also in the domain of $S$, we are again done.

We now deal with the induction step of $b)$. 
When $\deg(a)$ is odd then $C(a) = 1$. If the root is a sum gate then one of its children
evaluates to 0, and the other one 
to 1, say ${C_{\ell}}(a) = 0$ and ${C_{r}}(a) = 1$. 
By the inductive hypothesis there exists 
a unique $S' \in \cS(C_r)$ such that $m_{S'}(a) = 1$, and  $P_g(a) =  1$ for all $g \in \Dom(S')$.
On the other hand, if
$S \in \cS(C)$ such that $m_S(a) = 1$
and the mark of $S$ at the root is $\ell$, then $S_{\ell} \in \cS(C_{\ell})$ and $m_{S_{\ell}}(a) = 1$, and 
by $a)$ there exists $g \in \Dom(S)$ with $P_g(a) = 0$. 
Therefore the unique $S$ satisfying the
hypothesis of the statement is $S'$ extended with the mark $r$ at the root.

To finish the induction step for $b$), let us suppose now that the root of $C$ is a product gate. 
Then by the inductive hypothesis there exists
a unique $S' \in \cS(C_{\ell})$ such that $m_{S'}(a) = 1$, and  $P_g(a) =  1$ for all $g \in \Dom(S')$, and similarly
there exists a unique $S'' \in \cS(C_r)$ such that $m_{S''}(a) = 1$, and  $P_g(a) =  1$ for all $g \in \Dom(S'')$.
We claim that $S'$ and $S''$ are compatible, and therefore their union $S = S' \cup S''$ is the unique parse subcircuit of $C$
satisfying the claim. Suppose that it is not the case, that is there exists $g \in \Dom(S') \cap \Dom(S'')$ such that
$S'(g) \neq S''(g)$. Since $P_g(a) =  1$, for one of its children, say for $g_{\ell}$, we have 
$P_{g_{\ell}}(a) =  0$, contradicting the inductive hypothesis
about the parse subcircuit in $\{S',S''\}$ which takes the value $\ell$ in $g$.
\end{proof}

We give now the pairing $\phi(a, \cdot)$ for $a \in \F_2^n$.
If $\deg(a)$ is even then
let $S \in \cS(C)$ be such that $m_S(a) = 1$.  By Lemma~\ref{whatever}
there exists a sum gate in the domain of $S$ where $P$ evaluates to 0. 
Let $g$ be in some topological 
ordering of the gates of $C$ the first sum gate
such that $P_g(a) = 0$, and suppose without loss of generality that $S(g) = \ell$.
Let $Z \in \cS(C_g)$ be the restriction of $S$  to $C_g$,
and we obviously have $m_{Z}(a) = m_{Z_{\ell}}(a) = 1.$
We claim that $P_{g_{\ell}}(a) = P_{g_r}(a) = 1.$ Indeed, if $P_{g_{\ell}}(a) = P_{g_r}(a) = 0,$
then by Lemma~\ref{whatever}, applied to $C_{g_{\ell}}$, there exists $g' \in \Dom(Z_{\ell})$ with $P_{g'}(a) = 0$, which contradicts
the choice of $g$. Therefore again by Lemma~\ref{whatever} there exists a  
unique $Z'' \in \cS(C_{g_r})$ such that $m_{Z''}(a) = 1$, and  $P_h(a) =  1$ for all $h \in \Dom(Z'')$.
We let $Z' \in \cS(C_{g})$ be the extension of $Z''$ with $Z'(g) = r$. Finally we define
$\phi(a, S)$ as the parse subcircuit $S'$ obtained from $S$ by exchanging $Z$ with $Z'$, that is
$S' = (S \setminus Z) \cup Z'$. It is clear that $m_{S'}(a) = 1$, and $\phi(a, S') = S$.

If $\deg(a)$ is odd then by Lemma~\ref{whatever} 
there exists a unique parse subcircuit $S$ 
such that $m_S(a) = 1$, and  $P_g(a) =  1$, for all $g \in \Dom(S)$. We set $\phi(a, S) = S.$
For all parse subcircuits $S$ such that $P_g(a) =  0$, for some $g \in \Dom(S)$, the construction
of $S' = \phi(a, S)$ is identical to the previous case. 

The finish the proof, observe that the vertices of odd degree in $V_N$ other than the standard leaf $T$ are the
elements $a \in \F_2^n$ such that $C(a) = 1$. 
Therefore the output of the reduction is a satisfying assignment $a$ for $C$.
\end{proof}

\section{$\PPA$-hardness}
\label{sec:hardness}

\begin{theorem}
\label{thm:hardness}
\textsc{PPA-Circuit Chevalley} is $\PPA$-hard.
\end{theorem}

\begin{proof}

We will reduce \textsc{Leaf} to \textsc{PPA-Circuit Chevalley}.  Let $(z,M,\omega)$
be an instance of \textsc{Leaf}, where 
$M$ defines the graph $G_z = (V_z, E_z)$ with $V_z = \{0,1\}^n$, for some polynomial function $n$ of $|z|$,
and $\omega$ is the standard leaf in $G_z$. We know that for every vertex $u$, 
$M(z,u)$ is a set of at most two vertices. Composing the standard simulation of polynomial time Turing machines
by polynomial size boolean circuits~\cite{Sip'97} with the obvious simulation of boolean circuits by arithmetic circuits,
there exist two $n$-variables, $n$-output 
polynomial size arithmetic circuits $D$ and $F$ with the following properties:
\begin{itemize}
\item
if $M(z,u) = \emptyset$ or $M(z,u) = \{u\}$ then $D(u) = F(u) = u$,
\item
if $M(z,u) = \{v\}$ or $M(z,u) = \{v,u\}$ with $v \neq u$ then $D(u) = v$ and $F(u) = u$,
\item
if $M(z,u) = \{v,w\}$ with $v \neq u \neq w$  then $D(u) = v$ and $F(u) = w$ (or vice versa).
\end{itemize}

Consider the $\PPA$-composition $C_{D,F}$ of $D$ and $F$. We claim that for every vertex $u$, the degree of $u$
in $G_z$ is odd if and only if $u$ is a satisfying assignment for $C_{D,F}$. This is equivalent to saying that 
the parity of the degree of $u$ is the same as the parity of the satisfied components of $C_{D,F}$.
The proof of this claim is straightforward, but somewhat tedious.
We distinguish three cases in the proof,
depending on the cardinality of $M(z,u) \setminus \{u\}$.
\begin{itemize}
\item
Case 1: $M(z,u) \setminus \{u\} = \emptyset$. Then $u$ is an isolated vertex, and all six components are satisfied.
\item
Case 2: $M(z,u) \setminus \{u\} = \{v\}$. \\
{\em a}) If $u \in M(z,v)$ then the degree of $u$ is one,
and $I_5 \diamond F_3 \circ F_4, I_6 \diamond F_5$ 
and exactly one of the two components
		$I_2 \diamond F_2 \circ D_2, I_3 \diamond D_3 \circ D_4$ 
		are satisfied. \\
{\em b}) If $u \not \in M(z,v)$ then  $u$ is an isolated vertex, and $I_5 \diamond F_3 \circ F_4$ and $ I_6 \diamond F_5$  are satisfied.
\item
Case 3: $M(z,u) \setminus \{u\} = \{v,w\}$. \\
{\em a}) If $u \in M(z,v) \cap M(z,w)$ then the degree of $u$ is two, and exactly
one of the two components $I_2 \diamond F_2 \circ D_2, I_3 \diamond D_3 \circ D_4$
and exactly one of the two components $I_1 \diamond D_1 \circ F_1, I_5 \diamond F_3 \circ F_4$  are satisfied.\\
{\em b})
If $u \in M(z,v)$ but $u \not \in  M(z,w)$ and say $D(u) = v$, then exactly one of the two components 
$I_2 \diamond F_2 \circ D_2, I_3 \diamond D_3 \circ D_4$ is satisfied. \\
{\em c})
Finally, if $u \not \in M(z,v) \cup M(z,w)$ then $u$ is an isolated vertex, and none of the components is satisfied.
\end{itemize}

\begin{figure}[H]
	\centering
	\begin{subfigure}[b]{0.3\textwidth}
		\centering
\begin{tikzpicture}
	\node[circle, inner sep = 2pt, minimum size = 2pt, draw] (u) at (0, 0) [label=above:$u$] {};
\end{tikzpicture}
		\subcaption{Case $1$}
	\end{subfigure}
	\quad
	\begin{subfigure}[b]{0.3\textwidth}
		\centering
\begin{tikzpicture}
	\node[circle, inner sep = 2pt, minimum size = 2pt, draw] (u) at (0, 0) [label=above:$u$] {};
	\node[circle, inner sep = 2pt, minimum size = 2pt, draw] (v) at (1, 0) [label=above:$v$] {};

	\draw (u) -- (v);
	\node at (1.5, 0) {$\cdots$};
\end{tikzpicture}
		\subcaption{Case $2$-a}
	\end{subfigure}
	\quad
	\begin{subfigure}[b]{0.3\textwidth}
		\centering
\begin{tikzpicture}
	\node[circle, inner sep = 2pt, minimum size = 2pt, draw] (u) at (0, 0) [label=above:$u$] {};
	\node[circle, inner sep = 2pt, minimum size = 2pt, draw] (v) at (1, 0) [label=above:$v$] {};

	\node at (1.5, 0) {$\cdots$};
\end{tikzpicture}
		\subcaption{Case $2$-b}
	\end{subfigure}

	\begin{subfigure}[b]{0.3\textwidth}
		\centering
\begin{tikzpicture}
	\node[circle, inner sep = 2pt, minimum size = 2pt, draw] (u) at (1, 0) [label=above:$u$] {};
	\node[circle, inner sep = 2pt, minimum size = 2pt, draw] (v) at (0, 0) [label=above:$v$] {};
	\node[circle, inner sep = 2pt, minimum size = 2pt, draw] (w) at (2, 0) [label=above:$w$] {};

	\draw (v) -- (u) -- (w);
	\node at (-0.5, 0) {$\cdots$};
	\node at (2.5, 0) {$\cdots$};
\end{tikzpicture}
		\subcaption{Case $3$-a}
	\end{subfigure}
	\quad
	\begin{subfigure}[b]{0.3\textwidth}
		\centering
\begin{tikzpicture}
	\node[circle, inner sep = 2pt, minimum size = 2pt, draw] (u) at (1, 0) [label=above:$u$] {};
	\node[circle, inner sep = 2pt, minimum size = 2pt, draw] (v) at (0, 0) [label=above:$v$] {};
	\node[circle, inner sep = 2pt, minimum size = 2pt, draw] (w) at (2, 0) [label=above:$w$] {};

	\draw (v) -- (u);
	\node at (-0.5, 0) {$\cdots$};
	\node at (2.5, 0) {$\cdots$};
\end{tikzpicture}
		\subcaption{Case $3$-b}
	\end{subfigure}
	\quad
	\begin{subfigure}[b]{0.3\textwidth}
		\centering
\begin{tikzpicture}
	\node[circle, inner sep = 2pt, minimum size = 2pt, draw] (u) at (1, 0) [label=above:$u$] {};
	\node[circle, inner sep = 2pt, minimum size = 2pt, draw] (v) at (0, 0) [label=above:$v$] {};
	\node[circle, inner sep = 2pt, minimum size = 2pt, draw] (w) at (2, 0) [label=above:$w$] {};

	\node at (-0.5, 0) {$\cdots$};
	\node at (2.5, 0) {$\cdots$};
\end{tikzpicture}
		\subcaption{Case $3$-c}
	\end{subfigure}
	\caption{The six cases of Theorem~\ref{thm:hardness}.}
\end{figure}
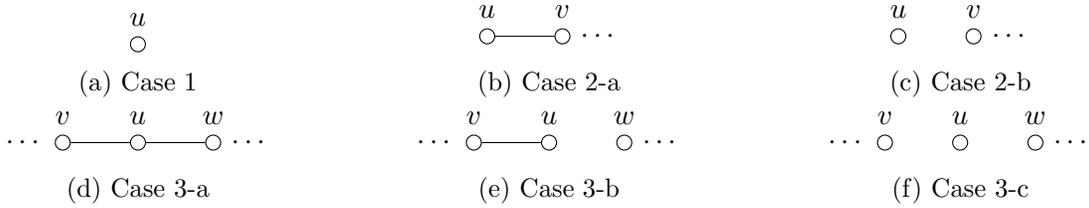

This finishes the proof of the claim. It follows that the number of satisfying assignments for $C_{D,F}$ is equal to the number
of leaves in $G_z$, which is even. The standard leaf $\omega$ is a satisfying 
assignment for $C_{D,F}$, and therefore
the output of \textsc{PPA-Circuit Chevalley} is another satisfying assignment, which is another leaf in $G_z$.
\end{proof}


\begin{thebibliography}{99}

\bibitem{ABB15}
J. Aisenberg, M. Bonet and S. Buss.
2--D Tucker is PPA-complete.
{\em ECCC Report} no. 163, 2015.

\bibitem{Alon99}
N. Alon.
Combinatorial Nullstellensatz.
{\em Combinatorics, Probability and Computing}, 8:1-2, pp. 7-29, 1999.

\bibitem{Alon02}
N. Alon.
\newblock Discrete Mathematics: Methods and Challenges.
\newblock {\em In Proc. of 
 $2002$ International Congress of Mathematicians (ICM)}, vol.~I,
pages 119--135, 2002.

\bibitem{BCEIP98}
P. Beame, S. Cook, J. Edmonds, R. Impagliazzo and T. Pitassi.
The relative complexity of NP search problems.
{\em Journal of Computer and System Sciences}, 57(1), pp. 3-19, 1998.

\bibitem{Ber86}
K. Berman.
Parity results on connected $f$-factors.
{\em Dicrete Math.}, 59, pp. 1-8, 1986.

\bibitem{BH86}
J. Bondy and F. Halberstam.
Parity theorems for cycles and cycles in graphs.
{\em J. of Graph Theory}, 10, pp. 107-115, 1986.

\bibitem{CE90}
K. Cameron and J. Edmonds.
Existentially poly-time theorems.
{\em DIMACS Series Discrete Mathematics and Theoretical Computer Science}, 1, pp. 83-99, 1990.

\bibitem{CE99}
K. Cameron and J. Edmonds.
Some graphic uses of an even number of odd nodes.
{\em Ann. Inst. Fourier} 49, pp. 1-13, 1999.

\bibitem{CD06}
X. Chen and X. Deng. 
On the complexity of 2D discrete fixed point problem.
{\em In Proc. of $33$rd ICALP},  pp. 489-500, 2006.

\bibitem{CDb06}
X. Chen and X. Deng.
Settling the complexity of two-player Nash equilibrium. 
{\em In Proc. of $47$th FOCS}, pp. 261--272, 2006.

\bibitem{CDT09}
X. Chen, X. Deng and S.-H. Teng. 
Settling the complexity of computing two-player Nash equilibria. 
{\em J. ACM}, 56(3) pp. 1-57, 2009.

\bibitem{Ch36}
C. Chevalley.
D\'emonstration d'une hypoth\`ese de M. Artin.
{\em Abhandlungen aus dem Mathematischen Seminar der Universit\"at Hamburg}, 11, pp. 73-75, 1936.

\bibitem{DEFLQX16}
X. Deng, J. Edmonds, Z. Feng, Z. Liu, Q. Qi and Z. Xu. 
Understanding PPA-completeness.
{\em in Proc. of $31$st CCC}, 
pp. 23:1-23:25, 2016.

\bibitem{Gri01}
M. Grigni.
\newblock A {S}perner lemma complete for {PPA}.
\newblock {\em Inform. Process. Lett.}, 77(5-6), pp. 255--259, 2001.

\bibitem{Jer16}
E. Je\v{r}\'abek. 
Integer factoring and modular square roots. 
{\em Journal of Computer and System Sciences}, 82, no. 2, pp. 380--394, 2016.

\bibitem{JS82}
Mark Jerrum and Marc Snir.
\newblock Some Exact Complexity Results for Straight-Line Computations over 
Semirings.
{\em J. ACM}, 29, no. 3, pp. 874--897, 1982.

\bibitem{KISV06}
K. Friedl, G. Ivanyos, M. Santha and Y. Verhoeven.
Locally 2-dimensional Sperner problem complete for the Polynomial Parity Argument classes.
{\em In Proc. $6$th CIAC}, pp. 380-391, 2006.

\bibitem{Kin}
S. Kintali. 
A compendium of $\PPAD$-complete problems. 
{\em http://www.cs.princeton.edu/kintali/ppad.html}.

\bibitem{Mal03}
G. Malod.
Polyn\^{o}mes et coefficients.
{\em Doctoral Thesis}, Universit\'e Claude Bernard, Lyon 1, 2003.

\bibitem{MP08}
G. Malod and N. Portier.
Characterizing Valiant's algebraic complexity classes.
{\em Journal of Complexity}, 24, pp. 16--38, 2008.

\bibitem{MP91}
N. Megiddo and C. Papadimitriou.
\newblock On total functions, existence theorems and computational complexity.
\newblock {\em Theoret. Comput. Sci.}, 81, pp. 317--324, 1991.



\bibitem{Pap90}
C. Papadimitriou.
\newblock On graph-theoretic lemmata and complexity classes. 
\newblock {\em In Proc. of $31$st FOCS}, pp. 794--801, 1990.


\bibitem{Pap94}
C. Papadimitriou.
\newblock On the complexity of the parity argument and other inefficient proofs
  of existence.
\newblock {\em J. Comput. System Sci.}, 48(3), pp. 498--532, 1994.

\bibitem{Sip'97}
M. Sipser.
Introduction to the Theory of Computation. 
PWS Publishing Company, 1997.

\bibitem{Toi73}
S. Toida.
Properties of an Euler graph.
{\em J. Franklin Institute}, 95, pp. 343-345, 1973.

\bibitem{Valiant05}
Leslie Valiant.
\newblock Completeness for parity problems.
\newblock In {\em International Computing and Combinatorics Conference}, pages
  1--8. Springer, 2005.

\bibitem{Varga14}
L. Varga.
Combinatorial Nullstellensatz modulo prime powers and the Parity Argument.
{\em Electr. J. Comb.}, 21(4), P4.44, 2014.

\bibitem{W36}
E. Warning.
Bemerkung zur vorstehenden Arbeit von Herrn Chevalley.
{\em Abhandlungen aus dem Mathematischen Seminar der Universit\"at Hamburg}, 11, pp. 76-83, 1936.

\bibitem{Wes78}
D. West.
Pairs of adjacent Hamiltonian circuits with small intersection.
{\em Studies of Applied Math.}, 59, pp. 245-248, 1978.

\end{thebibliography}

\section{Acknowledgments}

This research was partially funded by the Singapore Ministry of Education and the
National Research Foundation, also through the Tier 3 Grant ``Random numbers
from quantum processes,'' MOE2012-T3-1-009.  The research was also supported by
the ERC Advanced Grant MQC, 
the French ANR Blanc program under contract ANR-12-BS02-005 (RDAM project), and the
Hungarian National Research, Development and Innovation Office -- NKFIH Grant K115288.

Part of this work was performed when A. B., M. S. and S. Y. attended the program 
``Semidefinite and Matrix Methods for Optimization and Communication" hosted at the Institute for Mathematical 
Sciences, Singapore. We thank the Institute for the hospitality.
G.~I.~is grateful to the Centre for Quantum Technologies, NUS where part of his 
research was accomplished.

We are very grateful to several anonymous referees for numerous insightful comments on the paper.
We would also like to thank Herv\'e Fournier, Guillaume Malod and Sylvain Perifel for several helpful conversations.


\end{document}